\documentclass[runningheads,envcountsame,envcountsect]{llncs}

\usepackage{a4wide}

\usepackage[T1]{fontenc}
\usepackage{graphicx}

\usepackage{amsthm}
\usepackage{amsfonts}
\usepackage{amssymb}
\usepackage{amsmath}
\usepackage{multirow}
\usepackage{multicol}
\usepackage{algorithm}
\usepackage{algpseudocode}
\usepackage{etoolbox}
\usepackage{chngcntr}
\usepackage{framed}
\usepackage{latexsym}
\usepackage{hyperref}
\usepackage{url}
\usepackage{enumitem}
\usepackage{breakcites}
\usepackage{stmaryrd}
\usepackage{pgfplots}
\usepackage{pgfplotstable}
\usepackage{tikz}
\usetikzlibrary{patterns}
\usepackage[scr=boondox, scrscaled=1.05]{mathalfa}
\usetikzlibrary{tikzmark}
\usetikzlibrary{calc}

\definecolor{myGreen}{rgb}{0.0, 0.3, 0.0}

\usepackage{color}

\counterwithin{equation}{section}
\counterwithin{figure}{section}
\counterwithin{table}{section}
\counterwithin{algorithm}{section}

\spnewtheorem{notation}[theorem]{Notation}{\bfseries}{\itshape}
\spnewtheorem{fact}[theorem]{Fact}{\bfseries}{\itshape}
\spnewtheorem{pb}[theorem]{Problem}{\bfseries}{\itshape}
\spnewtheorem{conj}[theorem]{Conjecture}{\bfseries}{\itshape}
\spnewtheorem{assumption}[theorem]{Assumption}{\bfseries}{\itshape}

\usepackage{thm-restate}

\DeclareMathOperator*{\bias}{bias}

\newcommand{\eqdef}{\mathop{=}\limits^{\triangle}}

\newcommand{\transp}[1]{#1^\intercal}

\makeatletter
\newcommand*{\transpo}{{\mathpalette\@transpose{}}}
\newcommand*{\@transpose}[2]{\raisebox{\depth}{$\m@th#1\intercal$}}
\makeatother

\newcommand{\F}{\mathbb{F}}
\newcommand{\Ft}{\mathbb{F}_2}

\newcommand{\RR}{\mathbb{R}}

\newcommand{\NN}{\mathbb{N}}

\newcommand{\C}{\mathscr{C}}
\newcommand{\CC}{\C}
\newcommand{\DD}{\ensuremath{\mathscr{D}}}
\newcommand{\IInt}[2]{\left\llbracket #1, #2 \right\rrbracket}
\newcommand{\Iint}[2]{\IInt{#1}{#2}}

\newcommand{\Ctime}{\mathrm{T}}
\newcommand{\Cspace}{\mathrm{S}}
\newcommand{\CtimeD}[1]{\mathrm{T}_{\text{Dumer#1}}}
\newcommand{\CspaceD}[1]{\mathrm{S}_{\text{Dumer#1}}}
\newcommand{\alphaD}[1]{\alpha_{\text{Dumer#1}}}

\newcommand{\sA}{{\mathscr{A}}}
\newcommand{\sB}{{\mathscr{B}}}

\newcommand{\sH}{{\mathscr{H}}}
\newcommand{\sI}{{\mathscr{I}}}

\newcommand{\sL}{{\mathscr{L}}}
\newcommand{\sN}{{\mathscr{N}}}
\newcommand{\sP}{{\mathscr{P}}}

\newcommand{\sR}{{\mathscr{R}}}
\newcommand{\sS}{{\mathscr{S}}}
\newcommand{\sT}{{\mathscr{T}}}

\newcommand{\sY}{{\mathscr{Y}}}

\newcommand{\dist}[2]{\left| #1 - #2 \right|}
\newcommand{\scal}[2]{\left\langle #1 , #2 \right\rangle}

\newcommand{\hw}[1]{\left|#1\right|}

\newcommand{\prob}[2]{\mathbb{P}_{#1}\left(#2\right)}
\DeclareMathOperator{\Prob}{{\mathbb{P}}}

\newcommand{\var}[1]{{\textbf{Var}}\left( #1\right)}

\newcommand{\deltagv}{\delta_{\mathrm{GV}}}
\newcommand{\deltaGV}{\deltagv}

\newcommand{\OO}[1]{\ensuremath{\mathop{}\mathopen{}O\mathopen{}\left(#1\right)}}
\newcommand{\oo}[1]{\ensuremath{\mathop{}\mathopen{}o\mathopen{}\left(#1\right)}}
\newcommand{\om}[1]{\ensuremath{\mathop{}\mathopen{} \omega \mathopen{}\left(#1\right)}}
\newcommand{\Th}[1]{\Theta\left( #1 \right)}
\newcommand{\Tht}[1]{\widetilde{\Theta}\left( #1 \right)}
\newcommand{\OOt}[1]{\ensuremath{\mathop{}\mathopen{}\widetilde{O}\mathopen{}\left(#1\right)}}
\newcommand{\Om}[1]{\Omega \left( #1 \right)}
\newcommand{\Omt}[1]{\widetilde{\Omega} \left( #1 \right)}

\newcommand\dual[1]{#1^{\perp}}

\renewcommand{\vec}[1]{\mathbf{#1}}
\newcommand{\Gm}{{\mathbf{G}}}
\newcommand{\Hm}{{\mathbf{H}}}
\newcommand{\Am}{{\mathbf{A}}}
\newcommand{\Mm}{{\mathbf{A}}}
\newcommand{\Pm}{{\mathbf{P}}}

\newcommand{\Id}{{\mathbf{I}}}

\newcommand{\av}{{\mathbf{a}}}
\newcommand{\cv}{{\mathbf{c}}}

\newcommand{\ev}{{\mathbf{e}}}
\newcommand{\hv}{{\mathbf{h}}}
\newcommand{\sv}{{\mathbf{s}}}
\newcommand{\uv}{{\mathbf{u}}}

\newcommand{\xv}{{\mathbf{x}}}

\newcommand{\yv}{{\mathbf{y}}}

\newcommand{\Niter}{\text{N}_{\text{iter}}}

\newcommand{\Teq}{\text{T}_{\text{eq}}}
\newcommand{\Seq}{\text{S}_{\text{eq}}}
\newcommand{\Psucc}{\text{P}_{\text{succ}}}

\newcommand{\ie}{{\em i.e. }}



\makeatletter
\renewcommand\fnum@algorithm{\fname@algorithm~\thealgorithm.}
\makeatother

\algblockx[RepeatTimes]{RepeatTimes}{EndRepeat}
[1]{\textbf{repeat}\ #1\ \textbf{times}}
{\textbf{end repeat}}

\errorcontextlines\maxdimen

\newcommand{\ALGtikzmarkcolor}{black}\newcommand{\ALGtikzmarkextraindent}{4pt}\newcommand{\ALGtikzmarkverticaloffsetstart}{-.5ex}\newcommand{\ALGtikzmarkverticaloffsetend}{-.5ex}\makeatletter
\newcounter{ALG@tikzmark@tempcnta}

\newcommand\ALG@tikzmark@start{\global\let\ALG@tikzmark@last\ALG@tikzmark@starttext \expandafter\edef\csname ALG@tikzmark@\theALG@nested\endcsname{\theALG@tikzmark@tempcnta}\tikzmark{ALG@tikzmark@start@\csname ALG@tikzmark@\theALG@nested\endcsname}\addtocounter{ALG@tikzmark@tempcnta}{1}}

\def\ALG@tikzmark@starttext{start}
\newcommand\ALG@tikzmark@end{\ifx\ALG@tikzmark@last\ALG@tikzmark@starttext
\else
        \tikzmark{ALG@tikzmark@end@\csname ALG@tikzmark@\theALG@nested\endcsname}\tikz[overlay,remember picture] \draw[\ALGtikzmarkcolor] let \p{S}=($(pic cs:ALG@tikzmark@start@\csname ALG@tikzmark@\theALG@nested\endcsname)+(\ALGtikzmarkextraindent,\ALGtikzmarkverticaloffsetstart)$), \p{E}=($(pic cs:ALG@tikzmark@end@\csname ALG@tikzmark@\theALG@nested\endcsname)+(\ALGtikzmarkextraindent,\ALGtikzmarkverticaloffsetend)$) in (\x{S},\y{S})--(\x{S},\y{E});\fi
    \gdef\ALG@tikzmark@last{end}}

\apptocmd{\ALG@beginblock}{\ALG@tikzmark@start}{}{\errmessage{failed to patch}}
\pretocmd{\ALG@endblock}{\ALG@tikzmark@end}{}{\errmessage{failed to patch}}
\makeatother

\begin{document}
\title{Statistical Decoding 2.0: Reducing Decoding to LPN}
\author{K{\'e}vin Carrier\inst{1} \and
Thomas Debris-Alazard\inst{2} \and
Charles Meyer-Hilfiger\inst{3} \and Jean-Pierre Tillich\inst{3}}
\authorrunning{K. Carrier \and T. Debris-Alazard \and C. Meyer-Hilfiger \and J-P. Tillich}
\institute{
ETIS Laboratory, CY Cergy-Paris University, \email{kevin.carrier@ensea.fr}
\and
Project GRACE, Inria Saclay-Ile de France, \email{thomas.debris@inria.fr}
\and
Project COSMIQ, Inria de Paris, \email{charles.meyer-hilfiger@inria.fr,jean-pierre.tillich@inria.fr}
}
\maketitle              

\begin{abstract}
The security of code-based cryptography relies  primarily on the hardness of generic decoding with linear codes. The best generic decoding algorithms are all improvements of an old algorithm due to Prange:  they are known under the name of  information set decoders  (ISD).
A while ago, a generic decoding algorithm which does not belong to this family was proposed: statistical decoding. 
It is a randomized algorithm that requires the computation of a large set of parity-checks of moderate weight, and uses some kind of majority voting on these equations to recover the error. This algorithm was long forgotten because even the best variants of it 
performed poorly when compared to the simplest ISD algorithm.
 We revisit this old algorithm by using parity-check equations in a more general way. Here the parity-checks are used  to get LPN samples with a secret which is part of the error and the LPN noise is related to the weight of the parity-checks we produce. The corresponding LPN problem is then solved by standard Fourier techniques. By properly choosing the method of producing these low weight equations and the size of the LPN problem, we are able to outperform in this way significantly information set decoders at code rates smaller than $0.3$.  It gives for the first time after $60$ years, a better decoding algorithm for a significant range which does not belong to the ISD family. 
\end{abstract}

\section{Introduction}
\label{sec:intro}

\subsection{The Decoding  Problem and Code-based Cryptography}
Code-based cryptography relies crucially on the hardness of decoding generic linear codes which can be expressed as follows in the binary case
\begin{pb}[decoding a linear code]
\label{pb:decodingFq}
Let $\C$ be a binary linear code over $\Ft$ of dimension $k$ and length $n$, \ie a subspace of $\Ft^n$ of dimension $k$.
We are given $\yv \in \Ft^n$, an integer $t$ and want to find  a codeword $\cv \in \C$ and an error vector $\ev \in \Ft^n$ of Hamming weight $|\ev|=t$ for which $\yv = \cv+\ev$.
\end{pb}
This terminology stems from information theory, $\yv$ is a noisy version of a codeword $\cv$: $\yv=\cv+\ev$ where $\ev$ is a vector of weight $t$ and we want to recover the original codeword $\cv$. It can also be viewed as solving an underdetermined linear system with a weight constraint. Indeed, we can associate to a subspace $\C$ of dimension $k$ of $\Ft^n$ a binary $(n-k)\times n$ matrix $\Hm$ 
(also called a {\em parity-check} matrix of the code) whose kernel defines 
$\C$, namely $\C = \{\xv \in \Ft^n: \Hm \transp{\xv}=\mathbf{0}\}$. The decoding problem is equivalent to find an $\ev$ of Hamming weight $t$ such that
$\Hm \transp{\ev} = \transp{\sv}$ where $\sv$ is the {\em syndrome} of $\yv$ with respect to $\Hm$, \ie $\transp{\sv} = \Hm \transp{\yv}$.
This can be verified by observing that  if there exists $\cv \in \C$ and $\ev$ such that $\yv=\cv+\ev$ then $\Hm \transp{\yv} = \Hm \transp{(\cv+\ev)} = 
\Hm \transp{\cv}+\Hm \transp{\ev}= \Hm \transp{\ev}$.

The decoding problem has been studied for a long time and despite many efforts on this issue \cite{P62,S88,D91,B97b,FS09,BLP11,MMT11,BJMM12,MO15}
the best algorithms \cite{BJMM12,MO15,BM17, BM18} are exponential in the number of errors that have to be corrected:
correcting $t$ errors in a binary linear code of length $n$  with the aforementioned algorithms has a cost of $2^{\alpha n(1+ o(1)) }$ where $\alpha=\alpha(R,\tau)$ is a constant depending of the code rate $R \eqdef \frac{k}{n}$, the error rate 
$\tau \eqdef \frac{t}{n}$ and the algorithm which is used.
All the efforts that have been spent on this problem have only managed to decrease slightly this exponent $\alpha$.
Let us emphasize that this exponent is  the key for estimating the security level of any code-based cryptosystem.
We expect that this problem is the hardest at the Gilbert-Varshamov relative distance  $\tau=\deltagv$ where $\deltagv \eqdef h^{-1}(1-R)$, with $h$ being the binary entropy function $h(x) \eqdef - x \log_2 x -(1-x)\log_2(1-x)$ and $h^{-1}(x)$ its inverse ranging over $[0,\frac{1}{2}]$. This corresponds in the case of random linear codes to the largest relative weight below which there is typically just one solution of the decoding problem assuming that there is one. Above this bound, the number of solutions becomes exponential (at least as long as $\tau  < 1 - \deltagv$) and this helps to devise more efficient decoders. Furthermore, all the aforementioned algorithms become polynomial in the regime
$\frac{1-R}{2} \leq \tau \leq \frac{1+R}{2}$ (see an illustration of this behaviour in Figure \ref{fig:Prange}).
\begin{figure}[h!]
\centering
\includegraphics[height=6cm]{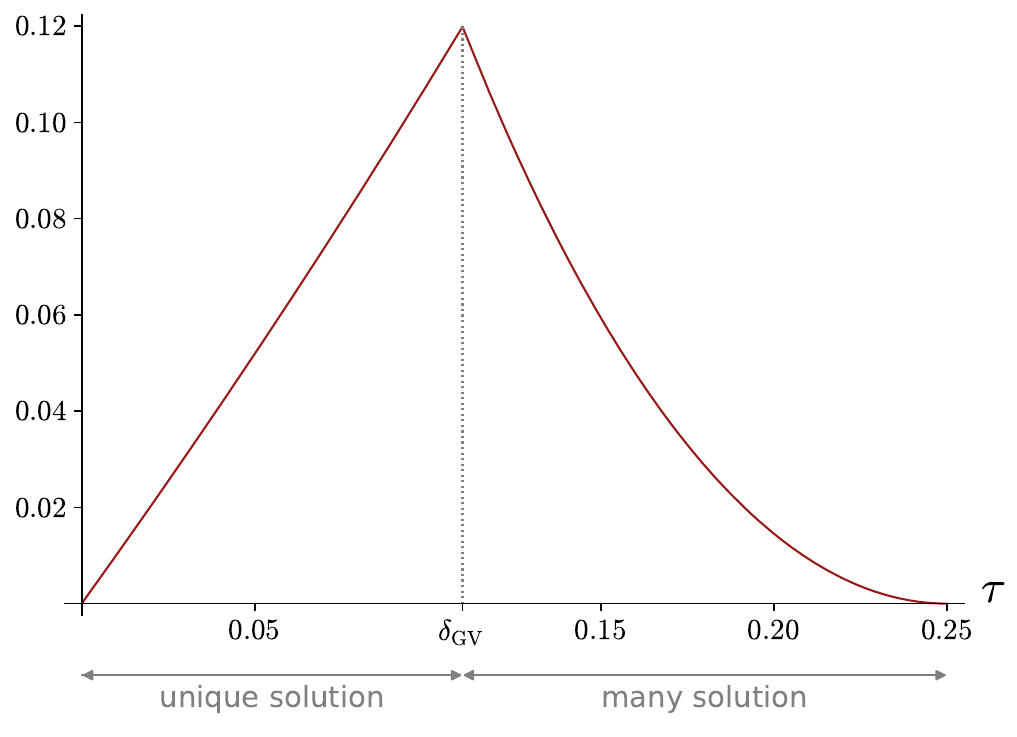}
\caption{\label{fig:Prange} Complexity exponent $\alpha$ 
of the Prange ISD algorithm \cite{P62} as a function of the error ratio $\tau \eqdef \frac{t}{n}$ at rate $R=\frac{1}{2}$. The peak corresponds to the normalized Gilbert-Varshamov distance $\delta_{\text{GV}}=h^{-1}(1-R)$.}
\end{figure}

There are code-based cryptographic primitives whose security relies precisely on the difficulty of decoding at the Gilbert-Varshamov relative distance (something which is also called {\em full distance decoding} \cite{MO15,BM17,BM18}), for instance the Stern code-based identification schemes or associated signatures schemes \cite{S93,GG07,AGS11,FJR21}. In the light of the upcoming NIST second call for new quantum resistant signature algorithms, it is even more important to have a stable and precise assessment of what we may expect about the complexity of solving this problem. For much smaller distances, say sub-linear, which is relevant for cryptosystems like \cite{MTSB13,M78}, the situation seems much more stable/well understood, since the complexity exponent of all the above-mentioned algorithms is the same in this regime \cite{CS16}.

\subsection{ISD Algorithms and Beyond: Statistical Decoding}
All the aforementioned algorithms can be viewed as a refinement of the original Prange algorithm \cite{P62} and are actually all 
referred to as Information Set Decoding (ISD) algorithms. Basically, they all use a common principle, namely making the bet that in a certain set of about $k$ positions (the ``information set'') there are only very few errors and using this bet to speed-up decoding. The parameters of virtually all code-based cryptographic algorithms (for the Hamming metric) have been chosen according to the running time of this family of algorithms. Apart from these algorithms, there is one algorithm which is worth mentioning, namely statistical decoding. 
It was first proposed by Al Jabri in \cite{J01} and improved a little bit by Overbeck in \cite{O06}. Later on,
\cite{FKI07} proposed an iterative version of this algorithm.

 It is essentially a two-stage algorithm, the first step 
consisting in computing an exponentially large number of parity-check equations of the
smallest possible weight $w$, and then from these parity-check equations the error is recovered by some 
kind of majority voting based on these equations. This majority voting is based on the following principle, take a parity-check equation $\hv$ for the code $\CC$ we want to decode, \ie a binary vector $\hv=(h_i)_{1 \leq i  \leq n}$ such that $\scal{\hv}{\cv}=0$ for every $\cv$ in $\CC$.
Assume that the $i$-th bit of the parity-check is $1$, then since $\scal{\hv}{\yv} = \scal{\hv}{\ev} = e_i + \sum_{j \neq i} h_j e_j$, the $i$-th bit $e_i$ of the error $\ev$ we want to recover satisfies
\begin{equation}\label{eq:stat}
e_i + \sum_{j \neq i} h_j e_j =  \scal{\hv}{\yv}.
\end{equation}
The sum $\sum_{j \neq i} h_j e_j$ is biased, say it is equal to $1$ with probability $\frac{1-\varepsilon}{2}$ with a bias $\varepsilon$ which is (essentially) a decreasing function of the weight $w$ of the  parity-check $\hv$. This allows to recover $e_i$ with about 
$\Th{1/\varepsilon^2}$  parity-checks. However the bias is exponentially small in the minimum weight of $\hv$ and $\ev$ and the complexity of such an algorithm is exponential in the codelength. An asymptotic analysis of this algorithm was performed in \cite{DT17} and it turns out that even if we had a way to obtain freely the parity-check equations we need, this kind of algorithm could not even outperform the simplest ISD algorithm: the Prange algorithm. This is done in \cite{DT17} by showing that there is no loss in generality if we just care about getting the best exponent to restrict ourselves to a single parity-check weight $w$ (see Section 5 in \cite{DT17}) and then analyse the complexity of such a putative algorithm for a single weight by using the knowledge of the typical number of parity-check equations of a given weight in a random linear code. The complexity exponent we get is a lower bound on the complexity of statistical decoding. We call such a putative statistical decoding algorithm, {\em genie-aided statistical decoding}: we are assisted by a genie which gives for free all the parity-check equations we require (but of course we can only get as much parity-check equations of some weight $w$ as there exists in the code we want to decode). The analysis of the exponent we obtain with such genie-aided statistical decoding  is given in \cite[\S 7]{DT17} and shows that it is outperformed very significantly by the  Prange algorithm (see \cite[\S 7.2, Fig. 6]{DT17}).

\subsection{Contributions}

In this work, we modify statistical decoding so that each parity-check yields now an LPN sample which is a noisy linear combination involving part of the error vector. This improves significantly statistical decoding, since the new decoding algorithm outperforms significantly all ISD's for code rates smaller than $0.3$. It gives for the first time after $60$ years, a better decoding algorithm that does not belong to the ISD family, and this for a very significant range of rates. The only other example where ISD algorithms have been beaten was in 1986, when Dumer introduced his collision technique. This improved the Prange decoder only for rates in the interval $[0.98,1]$ and interestingly enough it gave birth to all the modern improvements of ISD algorithms starting from Stern's algorithm \cite{S88}.
\newline

{\bf \noindent A New Approach : Using Parity-Checks to Reduce Decoding to LPN.}
Our approach for solving the decoding problem
reduces it to the so-called Learning Parity with Noise Problem (LPN). 
\begin{pb}[\textup{LPN}] \label{prob:LPN} 
	Let $\mathcal{O}_{\vec{s},\tau}(\cdot)$ be an oracle  parametrized by $\vec{s}\in\F_{2}^{s}$ and $\tau \in [0,1]$ such  that on a call it outputs $(\vec{a},\scal{\vec{s}}{\vec{a}} + e)$ where $\vec{a}\in\F_{2}^{s}$ is uniformly distributed and $e$ is distributed according to a Bernoulli of parameter $\tau$. We have access to $\mathcal{O}_{\vec{s},\tau}(\cdot)$ and want to find $\vec{s}$. 
\end{pb}

\eqref{eq:stat} can be interpreted as an LPN sample with an $\vec{s}$ of size $1$, namely $e_i$.
However, if instead of splitting the support of the parity-check with one bit on one side and the other ones on the other side, but choose say $s$ positions on the first part (say the $s$ first ones) and $n-s$ on the other, we can write
\begin{equation*}
\scal{\hv}{\yv} = \underbrace{\sum_{i=1}^s h_i e_i}_{\text{linear comb.}} + \underbrace{\sum_{j>s} h_j e_j}_{\text{LPN noise}}.
\end{equation*}
We may interpret such a scalar product as an LPN sample where the secret is $(e_1,\cdots,e_s)$; \ie we have a noisy information on a  linear combination $\sum_{i=1}^s h_i e_i$ on the $s$ first bits of the error where the noise is given by the term $\sum_{j>s} h_j e_j$ and the information is of the form
$\sum_{i=1}^s h_i e_i + \text{noise} =\scal{\hv}{\yv}$.  Again the second linear combination is biased, say 
$\Prob\left(\sum_{j>s} h_j e_j=1\right)=\frac{1-\varepsilon}{2}$ and information theoretic arguments show that again $\Th{1/\varepsilon^2}$ samples are enough to determine $(e_1,\cdots,e_s)$. It seemed that we gained nothing here since we still need as many samples as before 
and it seems that now recovering $(e_1,\cdots,e_s)$ is much more complicated than performing majority voting.

However with this new approach, we just need parity-check equations of low weight on $n-s$ positions (those that determine the LPN noise) whereas in  statistical decoding algorithm we have to compute parity-check equations of low weight on $n-1$ positions. 
This brings us to the main advantage of our new method: the parity-checks  we produce have much lower weight on those $n-s$ positions than those we produce for statistical decoding. This implies that the bias $\varepsilon$ in the LPN noise is much bigger with the new method and the number $N = \Th{1/\varepsilon^2}$ of parity-check equations much lower. Secondly, by using the fast Fourier transform, we can recover $(e_1, \cdots,e_s)$ in time $\OO{s2^s}$. Therefore, as long as the number of parity-checks we need  is of order $\Om{2^s}$, there is no exponential extra cost of having to recover $(e_1,\cdots,e_s)$. This new approach will be called from now on {\em Reduction to LPN} decoding (RLPN).  
\newline

{\bf \noindent Subset Sum Techniques  and Bet on the Error Distribution.} As just outlined, our RLPN decoder needs an exponential number $N = \Th{1/\varepsilon^2}$ of parity-checks of small weight on $n-s$ positions. 
This can be achieved efficiently by using collision/subset techniques used in the inner loop of ISD's. Recall that all  ISD's  proceed in two steps, $(i)$ first they pick an augmented information set and $(ii)$ then have an inner loop computing low weight codewords of some sort. Step $(ii)$ uses advanced techniques to solve subset-sum problems like birthday paradox \cite{D86,D91}, Wagner algorithm \cite{W02} or representations techniques \cite{MMT11,BJMM12}. All these techniques can also be used in a natural way in our RLPN decoder to compute the low weight parity-checks we need.

Furthermore, another idea of ISD's can be used in our RLPN decoder. All ISD's are making, in a fundamental way, a bet on the error weight distribution  in several zones related to the information set picked up in $(i)$. There are two zones: the potentially augmented information set  and the rest of the positions. ISD algorithms assume that the (augmented) information set contains only very few errors. A similar bet can be made in our case. We have two different zones: on one hand the $s$ positions determining $s$ error bits  and on the other $n-s$ bits which determine the LPN noise. It is clearly favourable to have an error ratio which is smaller on the second part. The probability that this unlikely event happens is largely outweighed by the gain in the bias of the LPN noise. 
\newline

{\bf \noindent Our Results.} Using all the aforementioned ingredients results in dramatically improving statistical decoding (see Figure \ref{fig:compISDRLPN}), especially in the low rate regime ($R  \leq \frac{1}{2}$) where ISD algorithms are known to perform slightly worse than in the 
high rate regime ($R > \frac{1}{2}$). 
Indeed, the complexity exponent $\alpha(R) \eqdef \alpha(R,\deltagv(R))$ of ISD's for full decoding ({\em a.k.a.} the GV bound decoding) which could be expected to be  symmetric in $R$ is actually bigger in the low rate regime than in the high rate regime: $\alpha(R) > \alpha(1-R)$ for $0 < R < \frac{1}{2}$. This results in an exponent curve which is slightly tilted towards the left, the maximum exponent being always obtained for $R < \frac{1}{2}$. Even worse, the behaviour for very small rates (\ie $R \rightarrow 0^+$) is fundamentally different in the very high rate regime ($R \rightarrow 1^-$). The complexity curve behaves like $\alpha(R) \approx R$ in the first case and like $\alpha(R) \approx \frac{1-R}{2}$ in the second (at least for all later improvements of the Prange decoder incorporating  collision techniques). This behaviour at $0$ for full distance decoding has never been changed by {\em any} decoder. It should be noted that $\alpha(R) = R(1+o(1))$ around $0$ means that the complexity behaves like
$2^{\alpha(R)n}=2^{R(1+o(1))n}=2^{k(1+o(1))}$, so in essence ISD's are not performing really better than trivial enumeration on all codewords. This fundamental barrier is still unbroken by our RLPN decoder, but it turns out that $\alpha(R)$ approaches $R$ much more slowly with  RLPN. For instance, for $R=0.02$ we have $\alpha(R) \approx \frac{R}{2}$. This behaviour in the very low regime is instrumental for the improvement we obtain on ISD's. In essence, this improvement is due in this regime to the conjunction of RLPN decoding with a collision search of low weight parity-checks. This method can be viewed as the dual (\ie operating on the dual code) of the collision search  performed in advanced ISD's  which are successful for lowering the complexity exponent down to $\alpha(R) \approx \frac{1-R}{2}$ in the high rate regime.
In some sense, the RLPN strategy allows us to {\em dualize} advanced ISD techniques for working in the low rate regime.

All in all, using \cite{BJMM12} (one of the most advanced ISD techniques)
to compute low weight codewords of some shape we are able to outperform significantly even the latest improvements of ISD algorithms for code rates $R$ smaller than $0.3$ as shown in Figure \ref{fig:compISDRLPN}. This is a breakthrough in this area, given the dominant role that ISD algorithms have played during all those years for assessing the complexity of decoding a linear code. Note however that the correctness of this algorithm relies on the LPN error model (Assumption \ref{ass:LPN}) for which some recent experiments have found out not to be completely accurate 
(see \url{https://github.com/tillich/RLPNdecoding/tree/master/verification_heuristic/histogram}). However, experimental results seem to indicate that this LPN modelling can be replaced by the weaker Conjecture \ref{con:main} which is compatible with the experiments we have made and for which there is a clear path to demonstrate its validity (see Subsection \ref{ss:assumption}).

\begin{figure}[h!]
	\centering
	\includegraphics[height=6.5cm]{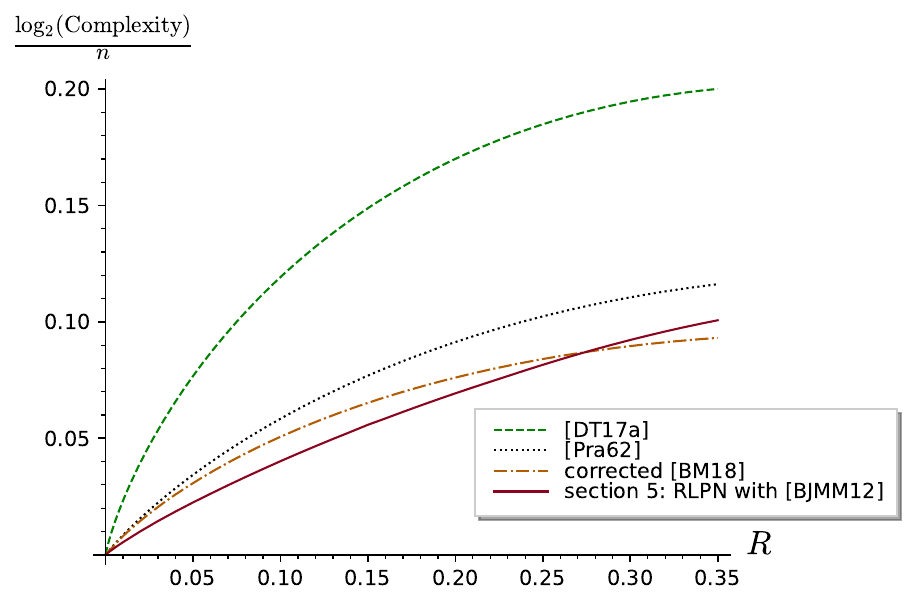}
	\caption{\label{fig:compISDRLPN} Complexity exponent for full distance decoding of genie-aided statistical decoding \cite[\S 7]{DT17} (recall that this is  a {\em lower bound} on the complexity exponent of statistical decoding), the basic Prange ISD algorithm \cite{P62}, the best state-of-the-art algorithm of \cite{BM18} (with a correction in the exponent that we give here, see Appendix \ref{sec:BM18}) and our RLPN decoder as a function of $R$. 
}
\end{figure}

\medskip
{\bf \noindent Proving the Standard Assumption of Statistical Decoding.} In analysing the new decoding algorithm, we also put statistical decoding on a much more rigorous foundation. We show that the basic condition that has to be met  for both statistical decoding and RLPN decoding, namely that the number $N$ of parity-check equations that are available is at least of order $\Om{1/\varepsilon^2}$ in the case of statistical decoding and $\Om{s/\varepsilon^2}$ in the case of RLPN decoding
where $\varepsilon$ is the bias of the LPN noise, is also essentially the condition which ensures that the bias is well approximated by the standard assumption made for statistical decoding which assumes that 
\begin{equation}
\label{eq:standardassumption}
\bias\left( \langle \vec{e}_{\sN},  \vec{h}_{\sN}\rangle \right)
	 \approx \bias\left(\langle \vec{e}_{\sN},  \vec{h'}_{\sN}\rangle \right),
	 \end{equation}
where $\bias(X)$ is defined for a binary random variable as $\bias(X) \eqdef \Prob(X=0)-\Prob(X=1)$, $\sN$ is a subset of $n-s$ positions (those which are involved in the LPN noise), $\hv$ is chosen uniformly at random
among the parity-checks of weight $w$ on $\sN$ of the code $\CC$ we decode whereas $\hv'$ is chosen uniformly at random among the words of weight $w$ on $\sN$. We will namely prove that as soon as the parameters are chosen such that 
 $N = \omega\left(1/ \bias\left(\langle \vec{e}_{\sN},  \vec{h'}_{\sN}\rangle \right)^2 \right)$, we have that for all but a proportion $o(1)$ of codes $\CC$ (as proved in Proposition  \ref{prop:bias} in Subsection \ref{ss:reductiontoLPN})
$$
\bias\left( \langle \vec{e}_{\sN},  \vec{h}_{\sN}\rangle \right) = (1+o(1)) \bias\left(\langle \vec{e}_{\sN},  \vec{h'}_{\sN}\rangle \right).
$$

 \section{Notation and Coding Theory Background}
\label{sec:notation}

In this section, we introduce notation and coding theoretic background  which are used throughout the paper.

\noindent
\paragraph{Vectors and matrices.}
Vectors and  matrices are respectively denoted in bold letters 
and bold capital letters such as $\av$ and $\Am$. The entry at index $i$ of the vector $\xv$ is denoted by $x_i$. The canonical scalar product $\sum_{i=1}^n x_i y_i$ between two vectors $\xv$ and $\yv$ of $\Ft^n$ is denoted by $\scal{\xv}{\yv}$. 
Let  $\sI$ be a list of indexes. We denote by $\xv_{\sI}$ the vector $(x_i)_{i \in \sI}$. In the same way, we denote by $\Mm_{\sI}$ the sub-matrix made of the columns of $\Mm$ which are indexed by $\sI$.
The concatenation of two vectors $\xv$ and $\yv$ is denoted by $\xv || \yv$. The Hamming weight of a vector $\vec{x}\in \F_{2}^{n}$ is defined as the number of its non-zero coordinates, namely
$
|\vec{x}| \eqdef \# \left\{ i \in \llbracket 1,n \rrbracket \mbox{ : } x_{i} \neq 0 \right\}
$
where $\#\sA$ stands for the cardinality of a finite set $\sA$ and $\IInt{a}{b}$ stands for the set of the integers between $a$ and $b$.
\newline 

\noindent
{\em Probabilistic notation.}
For a finite set $\sS$, we write $X \stackrel{\$}{\gets} \sS$ when $X$ is an element of $\sS$ drawn uniformly at random in it.
 For a Bernoulli random variable $X$, denote by $\bias(X)$ the quantity 
	$\bias(X) \eqdef \Prob(X=0)-\Prob(X=1)$. For a Bernoulli random variable $X$ of parameter $p=\frac{1-\varepsilon}{2}$, \ie $\Prob(X=1)=\frac{1- \varepsilon}{2}$, we have $\bias(X)=\varepsilon$. 
\newline 

\noindent
{\em Soft-O notation.}
For real valued functions defined over $\RR$ or $\NN$ we define $o()$, $\OO{}$, $\Om{}$, $\Th{}$, in the usual way and also use the less common notation $\OOt{}$ and $\Omt{}$, where 
$f = \OOt{g}$ means that $f(x) = \OO{g(x) \log^kg(x)}$ and $f=\Omt{g}$ means that $f(x)=\Om{g(x)\log^k g(x)}$ for some $k$. We will use this 
for functions which have an exponential behaviour, say $g(x)=e^{\alpha x}$, in which case $f(x)=\OOt{g(x)}$ means that 
$f(x)=\OO{P(x)g(x)}$ where $P$ is a polynomial in $x$.
We also use $f = \omega(g)$ when $f$ dominates $g$ asymptotically; that is when $\mathop{\lim}\limits_{x \rightarrow \infty} \frac{|f(x)|}{g(x)} = \infty$.
 
\noindent
\paragraph{Coding theory.} A binary linear code $\C$ of length $n$ and dimension $k$ is a subspace of the vector space $\Ft^n$ of dimension $k$. We say that it has parameters $[n, k]$ or that it is an $[n, k]$-code. Its {\em rate} $R$ is defined as $R \eqdef \frac{k}{n}$. 
A generator matrix $\Gm$ for $\C$ is a full rank $k \times n$ matrix over $\Ft$ such that
\begin{equation*}
\C =\left\{\uv \Gm:\uv \in \F_{2}^k \right\}.
\end{equation*}
In other words, the rows of $\Gm$ form a basis of $\C$. A parity-check matrix $\Hm$ for $\C$ is a full-rank $(n-k)\times n$ matrix over $\Ft$ such that
\begin{equation*}
\C = \left\{\cv \in \F_{2}^n: \Hm \transp{\cv} = \mathbf{0} \right\}.
\end{equation*}
In other words, $\C$ is the null space of $\Hm$.
The code whose generator matrix is the parity-check matrix of $\CC$ is called the dual code of $\CC$. It might be seen as the subspace of parity-checks of $\CC$ and is defined equivalently as
\begin{definition}[dual code]
The {\em dual code} $\CC^\perp$ of an $[n,k]$-code $\CC$ is an $[n,n-k]$-code which is defined by
$$
\CC^\perp \eqdef \left\{\vec{h} \in \Ft^n: \forall \cv \in \CC,\; \scal{\cv}{\vec{h}}=0 \right\}.
$$
\end{definition}
It will also be very convenient to consider the operation of puncturing a code, \ie keeping only a subset of entries in a codeword.
\begin{definition}[punctured code]
For a code $\CC$ and a subset $\sI$ of code positions, we denote by $\CC_\sI$ the punctured code obtained from $\CC$ by keeping only the positions in $\sI$, \ie
$$
\CC_{\sI}= \{\cv_\sI: \cv \in \CC\}.
$$
\end{definition}

We will also use several times that random binary linear codes can be decoded successfully, with a probability of error going to $0$, as the codelength goes to infinity as long as the code rate is below the capacity, and this of any binary input symmetric channel whose definition 
is 
\begin{definition}[binary input memoryless symmetric channel]
A binary input memoryless symmetric channel (BIMS) with output a finite alphabet $\sY$, is an error model on $\{0,1\}^*$  assuming that when a bit $b \in\{0,1\}$ is sent, it  gets mapped to $y \in \sY$ with probability denoted by $p(y|b)$ (these are  the transition probabilities of the channel). Being symmetric means that there is an involution $f$ such that 
$p(y|0)=p(f(y)|1)$.  Being memoryless means that the outputs of the channel are independent conditioned on the inputs: when $b_1\cdots b_n \in \{0,1\}^n$ is sent, the probability that the output is $y_1 \cdots y_n$ is given by $p(y_1|b_1) \cdots p(y_n|b_n)$.
\end{definition}
{We use here this rather general formulation to analyse what is going on when we have several different LPN samples corresponding to the same parity-check $\hv$. The error model that we have in this case will be more complicated than the standard binary symmetric channel (see Definition \ref{def:BSC} below).}
The capacity of such a channel is given by
\begin{definition}[capacity of a BIMS channel]
The capacity\footnote{The formula given here is strictly speaking the symmetric capacity of a channel, but these two notions coincide in the case of a BIMS channel.} $C$ of a BIMS channel with transition probabilities $(p(y|b))_{\substack{y \in \sY \\ b \in \{0,1\}}}$ is given by
$$
C \eqdef \sum_{y \in \sY}\sum_{b \in \{0,1\}} \frac{p(y|b)}{2} \log_2 \frac{p(y|b)}{\frac{1}{2} p(y|0)+\frac{1}{2}p(y|1)}.
$$
\end{definition}

{LPN samples correspond to  the binary symmetric channel (BSC) given by}
\begin{definition}[binary symmetric channel]
BSC$(p)$ is a BIMS channel with output alphabet $\sY=\{0,1\}$ and transition probabilities 
$p(0|0)=p(1|1)=1-p$, $p(1|0)=p(0|1)=p$, where $p$ is the crossover probability of the channel.
\end{definition}

In other words, this means that a bit $b$ is transformed into its opposite $1-b$ with probability $p$ when sent through the channel. It is readily verified that

\begin{definition}[binary symmetric channel]\label{def:BSC}
The capacity $C$ of BSC$(p)$ is given by 
$C=1-h(p)$.
\end{definition}

We will also talk about maximum likelihood decoding a code (under the assumption that the input codeword is chosen uniformly at random) for a given channel, meaning the following
\begin{definition}[maximum likelihood decoding]
Maximum likelihood decoding of a binary code $\CC \subset \{0,1\}^n$ over a BIMS channel with transitions probabilities $(p(y|b))_{\substack{y \in \sY \\ b \in \{0,1\}}}$ 
corresponds, given a received word $\yv \in \sY^n$, to output the (or one of them if there are several equally likely candidates) codeword $\xv$ 
which maximizes $p(\yv|\xv)$. Here
$p(\yv|\xv)\eqdef p(y_i|x_i)\cdots p(y_n|x_n)$ denotes the probability of receiving $\yv$ given that 
$\xv$ was sent. 
\end{definition}
In a sense, this is the best possible decoding algorithm for a given channel model. 
There is a variation of Shannon's theorem (see for instance \cite[Th. 4.68 p. 203]{RU08}) which says that a family of random binary linear codes $(\CC_n)_{n}$ attain the capacity of a BIMS channel.
\begin{theorem}\label{th:caplinearcodes}
Consider a BIMS channel of capacity $C$.
Let $\delta >0$ and consider a family of random binary linear codes $\CC_n$ of length $n$ and rate smaller than $(1-\delta) C$ obtained by choosing their generator matrix uniformly at random. Then under maximum likelihood decoding, the probability of error after decoding goes to $0$ as $n$ tends to infinity. 
\end{theorem} 
 \section{Reduction to LPN and the Associated Algorithm}\label{sec:algo}

The purpose of this section is $(i)$ to explain in detail the reduction to LPN, $(ii)$ to give a high level description of the algorithm which does not specify the method for finding the dual codewords we need, and then $(iii)$ to give its complexity. We assume from now on that we are given  $\yv$ which is equal to a sum of a codeword $\cv$ of the code $\C$ we want to decode plus an error vector $\ev$ of Hamming weight $t$:
$$
\yv = \cv + \ev, \;\;\cv \in \C,\;|\ev|=t.
$$
We will start this section by explaining how we reduce decoding to an LPN problem and also show how the LPN noise can be estimated accurately.

\subsection{Reduction to LPN}
\label{ss:reductiontoLPN}
Recall that in RLPN decoding we first randomly select a subset $\sP$ of $s$ positions
	$$
	\sP \subseteq \llbracket 1,n \rrbracket \quad \mbox{such that} \quad \#\sP= s
	$$
	where $s$ is a parameter that will be chosen later. $\sP$ corresponds to the entries of $\ev$ we aim to recover and is the secret  in the LPN problem. We denote by $\sN \eqdef \llbracket 1,n \rrbracket \setminus \sP$ the complementary set, with a choice of the letter $\sN$ standing for ``noise'' for reasons that will be clear soon. Given $\hv \in\dual{\CC}$, we compute,
	\begin{align*} 
	\langle \vec{y}, \vec{h} \rangle &= \langle \vec{e}, \vec{h} \rangle  \\ 
	&= \sum_{j \in \sP} h_{j}e_{j} + \sum_{j \in \sN} h_{j} e_{j} \\
	&= \langle \vec{e}_{\sP}, \vec{h}_{\sP} \rangle + \langle \vec{e}_{\sN}, \vec{h}_{\sN} \rangle 
	\end{align*} 
	It gives access to the following LPN sample:
	$$
\left( \vec{a},\langle \vec{s},\vec{a} \rangle + e  \right) \quad \mbox{where} \quad \left\{
	\begin{array}{l}
		\vec{s} \eqdef \vec{e}_{\sP}  \\
		\vec{a} \eqdef \vec{h}_{\sP} \\
		e \eqdef \langle \vec{e}_{\sN},  \vec{h}_{\sN} \rangle 
	\end{array}
	\right.
	$$
	Here $e$ follows a Bernoulli distribution that is a function of $n$, $s$ and $u$ ({\em resp. $w$}) the weight of $\vec{e}$ ({\em resp.} $\vec{h}$) restricted to $\sN$, namely	
$$
	u \eqdef |\vec{e}_{\sN}| \quad \mbox{and} \quad w \eqdef \left| \vec{h}_{\sN} \right|.
	$$ 
The probability that $e$ is equal to $1$ is estimated through the following proposition which gives for the first time a rigorous statement for  the standard assumption \eqref{eq:standardassumption} made for statistical decoding.
	 \begin{proposition}\label{prop:bias}
	 Assume that the code $\CC$ is chosen by picking for it an $(n-k) \times n$ binary parity-check matrix uniformly at random. 
	 Let $\sN$ be a fixed set of $n-s$ positions in $\IInt{1}{n}$ and $\ev$ be some error of weight $u$ on $\sN$. Choose $\hv$ uniformly at random among the parity-checks of $\CC$ of weight $w$ on $\sN$ and $\hv'$ uniformly at random among the words of weight 
	 $w$ on $\sN$. Let $\delta \eqdef \bias\left( \scal{\ev}{\hv'} \right)$. If the parameters $k$, $s$, $u$, $w$ are chosen as functions on $n$ so that  for $n$ going to infinity, the expected number $N$ of parity-checks of $\CC$ of weight $w$ on $\sN$ satisfies 
	 $N = \omega\left( 1 / \delta^2 \right)$ then for all but a proportion $o(1)$ of codes we have
	 $$
	 \bias\left( \langle \vec{e}_{\sN},  \vec{h}_{\sN}\rangle \right) = (1+o(1)) \delta.
	 $$
	 \end{proposition}
	 \begin{proof}
	 Let us define for $b \in \{0,1\}$:
	 \begin{eqnarray}
	 E_b \eqdef \#\{\hv \in  \CC^\perp:\; \;\left|\hv_{\sN}\right| = w, \;\scal{\ev_{\sN}}{\hv_{\sN}}=b\} \label{eq:Eb}\\
	 E'_b \eqdef \#\{\hv' \in  \Ft^n:\; \;\left|\hv'_\sN\right| = w, \;\scal{\ev_{\sN}}{\hv'_{\sN}}=b\} \label{eq:Epb}
	 \end{eqnarray}
	 By using \cite[Lemma 1.1 p.10]{B97b}\footnote{{Note that there is an additional condition ``Suppose $Lq^{-r}$ grows exponentially in $n$'' in the statement of this lemma, but it is readily seen that this condition is neither necessary nor used in the proof.}}, we obtain
	 \begin{eqnarray}
	 \mathbb{E}(E_b) & = & \frac{E'_b}{2^k}\\
	 \var{E_b} & \leq  &  \frac{E'_b}{2^k}
	 \end{eqnarray} 
	 By using now the Bienaym\'e-Tchebychev inequality, we obtain for any function $f$ mapping the positive integers to positive real numbers:
	 \begin{equation}
	 \prob{\CC}{\left|E_b- \mathbb{E}(E_b)\right|\geq  \sqrt{f(n) \mathbb{E}(E_b)}}  \leq  \frac{1}{f(n)}.
	 \end{equation}
	 Since $ \bias\left( \langle \vec{e}_{\sN},  \vec{h}_{\sN}\rangle \right)= \frac{E_0-E_1}{E_0+E_1}$ we have with probability greater than
	 $1 - \frac{2}{f(n)}$ that
	 \begin{equation}\label{eq:complicated}
	 \frac{\mu_0-\mu_1 - \sqrt{2f(n)} \sqrt{ \mu_0+\mu_1} }{\mu_0+\mu_1 + \sqrt{2f(n)} \sqrt{ \mu_0+\mu_1}}	 \leq \bias\left( \langle \vec{e}_{\sN},  \vec{h}_{\sN}\rangle \right) \leq  
	\frac{\mu_0-\mu_1 + \sqrt{2f(n)}\sqrt{ \mu_0+\mu_1} }{\mu_0+\mu_1 - \sqrt{2f(n)}\sqrt{ \mu_0+\mu_1}}	 \end{equation}
	 where $\mu_i \eqdef \mathbb{E}(E_i)$ and where we used that for all positive $x$ and $y$, $\sqrt{x}+\sqrt{y} \leq \sqrt{2(x+y)}$. We let 
	 $f(n) = \delta\sqrt{N}/2$. Since $N= \mu_0+\mu_1$ this implies $f(n) = \delta\sqrt{\mu_0+\mu_1}/2$. By the assumptions made in the proposition, note that $f(n)$ tends to infinity as 
	 $n$ tends to infinity.  We notice that 
	 \begin{eqnarray}
	 \sqrt{2f(n)} \sqrt{ \mu_0+\mu_1} &= & \delta^{1/2}(\mu_0+\mu_1)^{3/4} \nonumber \\
	 &= & o\left( \delta(\mu_0+\mu_1)\right) \label{eq:smallo}
	 \end{eqnarray}
	 because 
	 \begin{eqnarray*}
	 \frac{\delta^{1/2}(\mu_0+\mu_1)^{3/4}}{\delta(\mu_0+\mu_1)} 
	 &= & 	\frac{1}{\sqrt{\delta\sqrt{\mu_0+\mu_1}}} \\
	 & = & \frac{1}{\sqrt{2f(n)}} \\
	 & \rightarrow & 0 \text{ as $n \rightarrow \infty$}.
	 \end{eqnarray*}
	 Equation \eqref{eq:complicated} can now be rewritten as
	 \begin{equation}
	 \label{eq:simpler}
	  \frac{\mu_0-\mu_1 - o\left(\delta(\mu_0+\mu_1)\right) }{\mu_0+\mu_1 + o\left(\delta(\mu_0+\mu_1)\right) }	 \leq \bias\left( \langle \vec{e}_{\sN},  \vec{h}_{\sN}\rangle \right) \leq  
	\frac{\mu_0-\mu_1 + o\left(\delta(\mu_0+\mu_1)\right)  }{\mu_0+\mu_1 - o\left(\delta(\mu_0+\mu_1)\right)}	
	 \end{equation}
 Now, on the other hand
	 \begin{eqnarray*}
	 \delta & = & \bias\left( \scal{\ev_{\sN}}{\hv'_{\sN}} \right) = \frac{E'_0-E'_1}{E'_0+E'_1} = \frac{\frac{E'_0}{2^k}-\frac{E'_1}{2^k}}{\frac{E'_0}{2^k}+\frac{E'_1}{2^k}} \\
	 & = & \frac{\mu_0-\mu_1}{\mu_0+\mu_1} \;\;\text{(by \eqref{eq:Eb})}.
	 \end{eqnarray*}
	 From this it follows that we can rewrite \eqref{eq:simpler} as
	 \begin{equation}
	 \frac{\delta}{1+o(\delta)}- o(\delta) \leq \bias\left( \langle \vec{e}_{\sN},  \vec{h}_{\sN}\rangle \right)  \leq 
	  \frac{\delta}{1-o(\delta)}+ o(\delta) 
	 \end{equation}
	 from which it follows immediately that 
	 $
	 \bias\left( \langle \vec{e}_{\sN},  \vec{h}_{\sN}\rangle \right) = \delta(1+o(1)).
	 $
	 \end{proof}
	 \begin{remark}
	 Note that the condition $N = \Om{1/\delta^2}$, respectively $N = \Om{s/\delta^2}$ is the condition we need in order that statistical decoding, respectively RLPN decoding succeed. This means that if we just have slightly more equations than the ratio $\frac{1}{\delta^2}$, then the standard assumption  \eqref{eq:standardassumption} made for statistical decoding holds. The point of this assumption is that it allows easily to estimate 
	 the bias as the following lemma shows.
	 \end{remark}

	\begin{lemma}\label{lem:bias} 
		Under the same assumptions made in Proposition \ref{prop:bias} we have that for all but a proportion $o(1)$ of codes,
		$$
		\bias(\langle \vec{e}_{\sN},  \vec{h}_{\sN}\rangle) = \delta(1+o(1))\quad \mbox{with} \quad \delta \eqdef  \frac{ K_{w}^{n-s}(u)}{\binom{n-s}{w}}
		$$
		where $u \eqdef |\ev_{\sN}|$ and $K_w^{n}$ stands for  the Krawtchouk polynomial of order $n$ and degree $w\in \llbracket 0,n\rrbracket$ which is defined as:
	$$
	K_{w}^{n}(X) \eqdef \sum_{j=0}^{w} (-1)^j \binom{X}{j} \binom{n-X}{w-j}.
	$$ 
	\end{lemma}

	\begin{proof} By using Proposition \ref{prop:bias} (and the same notation as the one used there) we have that  for all but a proportion $o(1)$ of codes
	$$
	\bias\left( \langle \vec{e}_{\sN},  \vec{h}_{\sN}\rangle\right) = (1+o(1)) \bias\left( \langle \vec{e}_{\sN},  \vec{h}'_{\sN}\rangle\right).
	$$
	Now by definition of $u$ we have
		\begin{align*}
			\bias\left( \langle \vec{e}_{\sN},  \vec{h}'_{\sN}\rangle\right) &= 
			\frac{1}{\binom{n-s}{w}} \sum_{j \text{ even}} \binom{u}{j}\binom{n-s-u}{w-j} 
- \frac{1}{\binom{n-s}{w}} \sum_{j \text{ odd}} \binom{u}{j}\binom{n-s-u}{w-j} \\
&= \frac{1}{\binom{n-s}{w}} \sum_{j} (-1)^{j} \binom{u}{j}\binom{n-s-u}{w-j}\\
&=  \frac{ K_{w}^{n-s}(u)}{\binom{n-s}{w}}.
		\end{align*}
		
\end{proof}

We will now repeatedly denote by {\em bias} of the LPN sample the quantity $\varepsilon$ appearing in the previous lemma and the {\em estimated bias} the quantity namely
\begin{definition}[bias of the LPN samples]
The {\em bias} $\varepsilon$ of the LPN samples is defined by
$$
\varepsilon \eqdef \bias(\scal{\ev_{\sN}}{\hv_{\sN}}) 
$$
when $\vec{e}_{\sN}$ has Hamming weight $u$ and $\hv$ is drawn uniformly at random among the parity-check equations of weight $w$ restricted on $\sN$. The {\em estimated bias} is the quantity $\delta$ defined by
$$
\delta \eqdef \bias(\scal{\ev_{\sN}}{\hv'_{\sN}})
$$ 
when $\vec{e}_{\sN}$ has Hamming weight $u$ and $\hv'$ is drawn uniformly at random among the binary words of weight $w$ restricted on $\sN$. This quantity is equal to
$$
\delta = \frac{ K_{w}^{n-s}(u)}{\binom{n-s}{w}}.
$$
\end{definition}

The point of introducing Krawtchouk polynomials is that we can bring in asymptotic expansions of Krawtchouk polynomials. Most of the relevant properties we need about Krawtchouk polynomials are given in \cite[\S II.B]{KS21}. They can be summarized by

\begin{proposition}\label{prop:expansion}\mbox{ }
\begin{enumerate}
\item {\em Value at 0.} For all $0 \leq w \leq n$, $K_w^n(0)=\binom{n}{w}$.
\item {\em Reciprocity.} For all $0 \leq t,w \leq n$, $\binom{n}{t}K_w^n(t)=\binom{n}{w}K_t^n(w)$.
\item {\em Roots.} The polynomials $K_w^n$ have $w$ distinct roots which lie in the interval 
$$\Iint{n/2-\sqrt{w(n-w)}}{n/2+\sqrt{w(n-w)}}.$$ The distance between roots is at least $2$ and at most $o(n)$.
\item {\em Magnitude outside the root region.} We set $\tau \eqdef \frac{t}{n}$, $\omega \eqdef \frac{w}{n}$.
We assume $w \leq n/2$ and $t \leq n/2-\sqrt{w(n-w)}$. Let $z \eqdef \frac{1-2\tau  - \sqrt{D}}{2 (1-\omega)}$ where
	 $D \eqdef \left(1- 2 \tau \right)^2-4 \omega(1-\omega)$.
	 We have
	 \begin{equation}			\label{eq:realexplicit}
	 K^n_w(t)= 2^{n \left(\tau \log_2(1-z)+ (1-\tau) \log_2(1+z) - \omega\log_2 z +o(1)\right)}.
	 \end{equation}
	 \item {\em Magnitude in the root region.}
	Between any two consecutive roots of $K^n_{w}$, where $1 \leq w \leq \frac{n}{2}$, there exists $t$ such that:
		\begin{equation}\label{compexplicit}
		K_w^n(t) = 2^{n \left( \frac{1+h(\omega)-h(\tau)}{2}+o(1)\right)}
\;\;\text{where $\omega \eqdef \frac{w}{n}$ and $\tau \eqdef \frac{t}{n}$.}
		\end{equation} 
\end{enumerate}
\end{proposition}
By using this proposition, we readily obtain
\begin{proposition}[exponential behavior of $\delta^2$]\label{prop:bias0}
Let $\tau$ and $\omega$ be two reals in the interval $\left[0,\frac{1}{2}\right]$. Let $\omega^\perp \eqdef \frac{1}{2}-\sqrt{\omega(1-\omega)}$ and $z \eqdef \frac{1-2\tau  - \sqrt{D}}{2 (1-\omega)}$ where
	 $D \eqdef \left(1- 2 \tau \right)^2-4 \omega(1-\omega)$. There exists a sequence of positive integers 
$(t_n)_{n \in \NN}$ and $(w_n)_{n \in \NN}$, such that $\frac{t_n}{n} \mathop{\rightarrow}\limits_{n \rightarrow \infty} \tau$, 
$\frac{w_n}{n} \mathop{\rightarrow}\limits_{n \rightarrow \infty} \omega$ and 
$\frac{\log_2(K_{w_n}^n(t_n)^2/\binom{n}{w_n}^2)}{n}$ has a limit which we denote $\tilde{\delta}(\tau,\omega)$ with
$$
\widetilde{\delta}(\tau,\omega) = \left\{
\begin{array}{ll}
	2 \left(\tau \log_2(1-z)+ (1-\tau) \log_2(1+z) - \omega\log_2 z - h(\omega)\right)  & \mbox{ if } \tau \in [0,\omega^\perp] \\
	 1 - h(\tau)-h(\omega) & \mbox{ otherwise.}
\end{array}
\right.
$$
\end{proposition}
\begin{proof}
In the case $\tau \in [0,\omega^\perp]$ we just let $t_n = \lceil \tau n\rceil$, $w_n = \lceil \omega n\rceil$ and use directly the asymptotic expansion 
\eqref{eq:realexplicit}. In the case $\tau\in\left[\omega^\perp,\frac{1}{2}\right]$ we still define $w_n$ with $w_n \eqdef \lceil \omega n\rceil$ but define $t_n$ differently. For $n$ large enough, we know from Proposition \ref{prop:expansion} that $\lceil \tau n\rceil$ lies between two zeros of the Krawtchouk polynomial and that there exists an integer $t_n$ in this interval such that 
$
\frac{\log_2(K_{w_n}^n(t_n))}{n} = \frac{1+h(\omega)-h(\tau_n)}{2} + o(1)
$
where $\tau_n = \frac{t_n}{n}$. Now since the size of this interval is an $o(n)$ we necessarily have $\tau_n=\tau+o(1)$ and therefore 
$
\frac{\log_2(K_{w_n}^n(t_n))}{n} = \frac{1+h(\omega)-h(\tau)}{2} + o(1).
$
\end{proof}

The point of this proposition is that the term $2 \log_2(K_w^{n-s}(u)/\binom{n-s}{w})$ quantifies the exponential behaviour of the square $\varepsilon^2$ of the bias $\varepsilon$ (see Lemma \ref{lem:bias}) and $1/\varepsilon^2$  is up to polynomial terms the number of parity-checks we need for having enough information to solve the LPN problem as will be seen. This is because the capacity of the BSC$(\frac{1 - \varepsilon}{2})$ is $1 - h\left( \frac{1 - \varepsilon}{2}\right)=\theta(\varepsilon^2)$ and that solving an LPN-problem with a secret of size $s$ and $N$ samples amounts to be able to decode a random linear code of rate $\frac{s}{N}$ over the BSC$(\frac{1 - \varepsilon}{2})$. It is therefore doable as soon as the rate is below the capacity (see Theorem \ref{th:caplinearcodes}). The reason why the Shannon capacity appears here is because of the following heuristic/assumption we will make here:
\begin{assumption}[LPN modelling]\label{ass:LPN}
We will assume that the $\scal{\ev_\sN}{\hv_\sN}$ are i.i.d Bernoulli random variables of parameter $\frac{1-\varepsilon}{2}$.
\end{assumption}

Strictly speaking, the corresponding random variables are not independent. {However, note that similar heuristics are also used to analyze a related lattice decoder making use of short dual lattice vectors (they are called dual attacks in the literature). We will discuss this assumption in more depth in Subsection \ref{ss:assumption}. Assumption \ref{ass:LPN} models the LPN noise as a binary symmetric channel BSC$(\frac{1- \varepsilon}{2})$ of crossover probability $\frac{1- \varepsilon}{2}$. A straightforward application of Theorem \ref{th:caplinearcodes} together with the fact that the capacity of a binary symmetric BSC$(\frac{1- \varepsilon}{2})$ is $1-h\left( \frac{1- \varepsilon}{2}\right)=\Omega(\varepsilon^2)$ implies}

\begin{fact} With Assumption \ref{ass:LPN}, the number $N$ of LPN samples is such that $s/N = O(\varepsilon^2)$ for a small enough constant in the $O$, performing maximum-likelihood decoding of the corresponding $[N,s]$ binary code recovers the secret $\ev_\sP$ with probability $1-o(1)$.
\end{fact}

Performing maximum likelihood decoding of the corresponding code can be achieved by a fast Fourier transform on a relevant vector. Indeed, for a given received word $\yv$ and a set $\widetilde{\sH}$ of $N$ parity-checks so that their restriction to $\sP$ leads to a set $\sH$ of $N$ different vectors of $\Ft^s$, we let for $\av \in \sH$, $\widetilde{\av}$ be the unique parity-check in $\widetilde{\sH}$ such that $\widetilde{\av}_{\sP}=\av$ and define  $f_{\yv,\sH}$ as 
\begin{eqnarray}
f_{\yv,\sH}: \Ft^s & \rightarrow & \RR \nonumber \\
\av & \mapsto & \left\lbrace \begin{array}{ll} (-1)^{\scal{\yv}{\widetilde{\av}}}\; & \text{ if $\av \in \sH$} \\ 0 & \text{ otherwise} \end{array} \right.\label{eq:def_f_H}
\end{eqnarray}	
We define the Fourier transform of such a function by
$$
\widehat{f}(\xv) \eqdef \sum_{\uv \in \Ft^s} f(\uv) (-1)^{\scal{\xv}{\uv}}.
$$
The code $\DD$ we want to decode (obtained via our LPN samples) is described as
\begin{equation}
	\DD \eqdef  \{\cv_{\xv},\;\xv \in \Ft^s \} \;\;\text{where}  \;\;
\cv_{\xv} \eqdef  (\scal{\xv}{\av})_{\av \in \sH}, \label{eq:def_C} 
\end{equation}
and the word $\uv_{\yv,\sH}$ we want to decode is given by
$
\uv_{\yv,\sH}=(\scal{\yv}{\widetilde{\av}})_{\av \in \sH}.
$
It is readily seen that
\begin{eqnarray*}
\widehat{f_{\yv,\sH}}(\xv) & = & \sum_{\av \in \Ft^s} f(\av) (-1)^{\scal{\xv}{\av}}\\
& = & \sum_{\substack{\av \in \sH}} (-1)^{\scal{\xv}{\av}+\scal{\yv}{\tilde{\av}}}\\
& = & \#\sH - 2 |\uv_{\yv,\sH}-\cv_{\xv}|.
\end{eqnarray*}
In other words, finding the closest codeword to $\uv_{\yv,\sH}$ is nothing but finding the $\xv$ which maximizes $\widehat{f_{\yv,\sH}}(\xv)$. 
This is achieved in time $\OO{s2^s}$ by performing a fast Fourier transform. Notice that an exhaustive search would cost $\OO{2^{2s}}$.

\subsection{Sketch of the whole algorithm}

\begin{algorithm}
\caption{RLPN decoder \label{alg:RLPN}}
\hspace*{\algorithmicindent} \textbf{Input:} $\yv$, $t$, $\CC$ an $\lbrack n,k \rbrack$-code\\ 
 \hspace*{\algorithmicindent} \textbf{Output:} $\ev$ such that $|\ev|=t$ and $\yv - \ev \in \CC$.
\begin{algorithmic}
\Function{RLPNdecode}{$\yv$, $\CC$, $t$}
\State $s, u \gets \Call{Optim}{t,\;k,\;n}$ 
\State
\Comment{$s$ and $u$ in order to minimize the complexity of the following procedure.}
\For{$i$ from $1$ to $\Niter$}
\Comment{$\Niter$ is a certain function of $n$, $s$, $t$ and $u$.}
\State $\sP \stackrel{\$}{\gets} \left\{\sI \subseteq \llbracket 1, n \rrbracket \ : \ \#\sI = s\right\}$
\State $\sN \gets \llbracket 1, n \rrbracket \setminus \sP$
\State $\sH \gets$ \Call{Create}{$N,w,\sP$}
\State $\widehat{f_{\yv,\sH}}\gets$\Call{FFT}{$f_{\yv,\sH}$}
\State $\xv_0 \gets \arg\max \widehat{f_{\yv,\sH}}$
\If{$\widehat{f_{\yv,\sH}}(\xv_0) \geq \frac{\delta N}{2}$} \Comment{$\delta \eqdef  {K_{w}^{n-s}(u)}/{\binom{n-s}{w}}$}.
\State \Return $\ev$ such that $\ev_{\sP} = \xv_0$ and $\ev_{\sN} = \Call{RLPNdecode}{\yv_{\sN},\CC_{\sN},t - \hw{\xv_0}})$
\EndIf
\EndFor
\EndFunction
\end{algorithmic}
\end{algorithm}	

Besides, the fast Fourier transform solving the LPN problem, Algorithm \ref{alg:RLPN} uses two other ingredients:
\begin{itemize}
\item A routine \Call{Create}{$N,w,\sP$} creating a set $\sH$ of $N$ parity-check equations $\hv$ such that $|\hv_\sN|=w$ where $\sN \eqdef \llbracket 1, n \rrbracket \setminus \sP$.
We will not specify how this function is realized here: this is done in the following sections. This procedure together with an FFT for decoding the code associated to the parity-check equations in $\sH$ (see Equation \eqref{eq:def_C}) form the inner loop of our algorithm.
\item An outer loop making a certain number $\Niter$ of calls to the inner procedure, checking each time a new set $\sP$ of $s$ positions with the hope of finding an $\sN$ containing an unusually low number $u$ of  errors in it. The point is that with a right $u$, the number of times we will have to check a new  $\sP$ is outweighed by the decrease in $N$ because the bias $\delta$ is much higher for such a $u$. 
\end{itemize}

\subsection{Analysis of the RLPN decoder}	

We need to show now that our RLPN decoder returns what we expect.
\begin{proposition}[acceptation criteria]\label{prop:acceptation}
Under Assumption \ref{ass:LPN}, 
by choosing $N_{\mathrm{iter}} = \omega\left(\frac{1}{\Psucc}\right)$ (where $\Psucc$ is the probability over the choice of $\sN$ that there are exactly $u$ errors in $\sN$), $s = \omega(1)$ and $N = \omega\left( \frac{n}{\delta^2}\right)$,
we have with probability $1 - \oo{1}$ that at least one iteration is such that $\ev_{\sP}$ meets the acceptation criteria $\widehat{f_{\yv,\sH}}(\ev_{\sP}) \geq \frac{\delta N}{2}$. Moreover, the probability that there exists $\xv \neq \ev_{\sP}$ which meets this acceptation criteria is $o(1)$.
\end{proposition}

\begin{proof}
We need to show that two things happen both with probability $1-o(1)$: 
 $(i)$ there is at least one iteration in the Algorithm \ref{alg:RLPN} for which $\widehat{f_{\yv,\sH}}(\ev_{\sP}) \geq \frac{\delta N}{2}$ and
$(ii)$ for all $\xv \in \Ft^s$ different from $\ev_{\sP}$, we have $\widehat{f_{\yv,\sH}}(\xv) < \frac{\delta N}{2}$ for all iterations.\\
The first point $(i)$ follows from the fact that by taking $\Niter \eqdef \omega\left(\frac{1}{\Psucc}\right)$, we have that one iteration is such that $\hw{\ev_{\sN}} = u$ and $\hw{\ev_{\sP}} = t - u$ with probability $1 - \oo{1}$. For such an iteration, we have from Assumption \ref{ass:LPN} that $\dist{\cv_{\ev_{\sP}}}{\uv_{\yv,\sH}} \sim \mathrm{Binomial}\left(N, \frac{1 - \varepsilon}{2}\right)$. Thus, by using the Hoeffding inequality,
\begin{equation*}
	\Prob\left(\widehat{f_{\yv,\sH}}(\ev_{\sP}) \geq \frac{\delta N}{2} \right) = \Prob\left(\dist{\cv_{\ev_{\sP}}}{\uv_{\yv,\sH}} \leq \frac{1-\frac{\delta}{2}}{2} N\right) \\
	\geq  1 - \exp \left( - \frac{(\varepsilon-\delta/2)^2 N}{2}\right)
\end{equation*}
which is a $1-o(1)$ by the choice made on $N$.

\noindent For the second point $(ii)$, consider now an $\xv \in \Ft^s$ such that $\xv \neq \ev_{\sP}$. Let $\cv' \eqdef \cv_{\xv} - \cv_{\ev_{\sP}}$, $d \eqdef \hw{\cv'}$ and $\yv' \eqdef \uv_{\yv,\sH} - \cv_{\ev_{\sP}}$. Then we have:
\begin{eqnarray}
\Prob\left(\widehat{f_{\yv,\sH}}(\xv) \geq \frac{\delta N}{2} \right) & = & \Prob\left(\dist{\cv'}{\yv'} \leq \frac{1 - \frac{\delta}{2}}{2} N \right) \;\;\text{(because $|\cv_\xv-\uv_{\yv,\sH}|=|\cv'-\yv'|$)} \nonumber\\
& = & \sum_{b = 0}^{N} \sum_{a=0}^{\frac{1 - \frac{\delta}{2}}{2} N }\Prob\left( \dist{\cv'}{\yv'} = a \ \Big\vert \ \hw{\yv'} = b \right) \; \Prob\left( \hw{\yv'} = b \right) \nonumber \\
& \leq & \sum_{b = 0}^{N} \sum_{a=0}^{\frac{1 - \frac{\delta}{2}}{2} N } \Prob_{\yv'}\left( \dist{\cv'}{\yv'} = a \ \Big\vert \ \hw{\yv'} = b \right) \nonumber \\
& \leq & \sum_{b = 0}^{N} \sum_{a=0}^{\frac{1 - \frac{\delta}{2}}{2} N } \frac{\binom{b}{\frac{d+b-a}{2}} \binom{N - b}{\frac{d-b+a}{2}}}{\binom{N}{d}}. \label{eq:last}
\end{eqnarray}
The last point follows from the fact that $|\yv'| \sim \mathrm{Binomial}\left(N, \frac{1 - \varepsilon'}{2}\right)$ for some $\varepsilon'$ depending on $|\ev_{\sN}|$ and therefore the conditional distribution of $\yv'$ given that $|\yv'|=b$ is the uniform distribution of words of weight $b$ over $\Ft^N$. This implies
$$
 \Prob_{\yv'}\left( \dist{\cv'}{\yv'} = a \ \Big\vert \ \hw{\yv'} = b \right) = \frac{\binom{b}{\frac{d+b-a}{2}} \binom{N - b}{\frac{d-b+a}{2}}}{\binom{N}{d}}.
$$

Note that both $\cv_{\xv}$ and $\cv_{\ev_{\sP}}$ are in the code $\DD$ defined by Equation \eqref{eq:def_C}. We can bound again the (typical) minimum non-zero weight and the  maximum weight by using \cite[Lemma 1.1, p.10]{B97b}. 
Let $\sB = \{\xv \in \Ft^N: 0 < |\xv| \leq N h^{-1}(1-2s/N)\}$. The expected number $M$ of codewords of $\DD$ in $\sB$ is $\frac{\#\sB}{2^{N(1-s/N)}}$.
Since $\#\sB \leq 2^{Nh(h^{-1}(1-2s/N))}=2^{N(1-2s/N)}$ we have that $M \leq 2^{-s}$. Since the probability that the minimum distance of
$\DD$ is less than or equal to $N h^{-1}(1-2s/N)$ is upper-bounded by $M$, we obtain that the minimum distance of $\DD$ is greater than
$N h^{-1}(1-2s/N)$ with probability $1-o(1)$. A similar reasoning can be made for the maximum weight. We therefore obtain that with probability
$1-o(1)$ all the weights $d$ of the non-zero codewords of $\DD$ lie in $\IInt{d^-}{d^+}$ where
\begin{equation*}
d^- \eqdef Nh^{-1}(1-2s/N) \quad \mbox{and} \quad 
d^{+} \eqdef N - Nh^{-1}(1-2s/N).
\end{equation*}
In such a case, we always have for $n$ large enough
\begin{eqnarray}
a \leq&  d  &\leq N-a \label{eq:d}\\
1-2s/N  \leq &h(d/N) & \leq 1.\label{eq:hdN}
\end{eqnarray}
We observe now that
\begin{eqnarray}
\frac{\binom{b}{\frac{d+b-a}{2}} \binom{N - b}{\frac{d-b+a}{2}}}{\binom{N}{d}} & \leq & 2^{b h\left( \frac{1}{2} - \frac{d - a}{2b}\right) + (N - b) h\left( \frac{1}{2} - \frac{N - d - a}{2(N-b)}\right) - N h\left( \frac{d}{N}\right)} \;\;\text{(because of \eqref{eq:d})} \nonumber\\
& \leq & 
2^{ \min \left[ - b \left( 1 - h\left( \frac{1}{2} - \frac{d - a}{2b}\right)\right) + 2s \ , \ -(N - b) \left( 1 - h\left( \frac{1}{2} - \frac{ d - a}{2(N-b)}\right)\right) + 2s \right] } \;\;\text{(because of \eqref{eq:hdN})} \nonumber\\
& \leq & 2^{-\frac{N}{2} \left( 1 - h\left( \frac{1}{2} - \frac{d - a}{2N}\right)\right) + 2s} \label{eq:anotherlast}
\end{eqnarray}
The last inequality comes from the fact that  either $b$ or $N-b$ is greater than $\frac{N}{2}$ and that both of them are smaller than $N$.
Since $h^{-1}(1-u) = \frac{1}{2}-\Th{\sqrt{u}}$ for $u \rightarrow 0^+$ we have that 
$ N/2(1- \Th{\sqrt{\frac{s}{N}}}) \leq d \leq N/2(1+ \Th{\sqrt{\frac{s}{N}}})$ and
therefore
\begin{equation}\label{eq:delta_final}
\frac{d-a}{2N} \geq \frac{\delta/2-\Th{\sqrt{\frac{s}{N}}}}{4} = \Om{\delta}.
\end{equation}
By using that $h(1/2-u)=1-\Th{u^2}$ for $u \rightarrow 0$ together with \eqref{eq:delta_final} in \eqref{eq:anotherlast} we obtain
\begin{equation}\label{eq:lastlast}
\frac{\binom{b}{\frac{d+b-a}{2}} \binom{N - b}{\frac{d-b+a}{2}}}{\binom{N}{d}}  \leq 2^{-\Om{N\delta^2}+2s}.
\end{equation}
By plugging this inequality in \eqref{eq:last} we finally obtain
$$
\Prob\left(\widehat{f_{\yv,\sH}}(\xv) \geq \frac{\delta N}{2} \right)  \leq N^2 \; 2^{-\Om{N\delta^2}+2s}
$$
and the probability of the event ``there exists an iteration and an $\xv \neq \ev_\sP$ such that $\widehat{f_{\yv,\sH}}(\xv) \geq \frac{\delta N}{2}$''
is upper-bounded by
$$
N_{\text{iter}} \; 2^s N^2 2^{-\Om{N\delta^2}+2s} = o(1).
$$
\end{proof}

	The space and time complexity of this method are readily seen to be given by 
\begin{proposition}\label{prop:complexity1}
Assume that \Call{Create}{$N$,$w$,$\sP$} produces $N$ parity-check equations in space $\Seq$ and time 
$\Teq$. The probability $\Psucc$ (over the choice of $\sN$) that there are exactly $u$ errors in $\sN$ is given by
$
\Psucc = \frac{\binom{s}{t-u}\binom{n-s}{u}}{\binom{n}{t}}.
$ The space complexity $S$ and the time complexity $T$ of the RLPN-decoder are given by 
$$
\textbf{\textup{Space:}}\;\; S=\OO{\Seq+2^s}, \quad \textbf{\textup{Time:}}\;\; T = \OOt{\frac{\Teq+2^s}{\Psucc}}.
$$
The parameters $s$, $u$ and $w$ have to meet the following constraints
\begin{eqnarray}
N & \leq & 2^s \label{eq:constraint_N}\\
N & \leq & \frac{\binom{n-s}{w}}{2^{k-s}}.  
\end{eqnarray}
Under Assumption \ref{ass:LPN} the algorithm outputs the correct $\ev_\sP$ with probability $1-o(1)$ if in addition we choose $N$ and 
$\Niter$ such that
\begin{eqnarray}
N &= &\omega\left(n \left(\frac{\binom{n-s}{w}}{K_w^{n-s}(u)}\right)^2 \right) \label{eq:constraint_N_bis}\\
\Niter& =& \omega\left(\frac{1}{\Psucc}\right). \label{eq:constraint_Niter}
\end{eqnarray}
\end{proposition}
\begin{proof}
All the points are straightforward here, with the exception of the constraints. The first constraint is that the number of parity-checks should not be bigger than the total number of different LPN samples we can possibly produce. The second one is that the number of parity-checks needed is smaller than the number of available parity-checks. 
The conditions ensuring the correctness of the algorithm follow immediately from Proposition \ref{prop:acceptation}.
\end{proof}

\subsection{{On the validity of Assumption \ref{ass:LPN}}}
\label{ss:assumption}

{
The proof of the correctness of the algorithm relies on the validity of the LPN modelling (Assumption \ref{ass:LPN}). We have programmed this algorithm and have verified that for several parameters it gives the correct answer. The corresponding experiments with the programs that have been used for running them can be found on \url{https://github.com/tillich/RLPNdecoding}. However, we have also found out (see \url{https://github.com/tillich/RLPNdecoding/tree/master/verification_heuristic/histogram}) that the second largest Fourier coefficient (the one which corresponds to the second nearest codeword, besides $\ev_{\sP}$) does not behave in the same way in the LPN model as in practice with the noise given by the $\scal{\hv_{\sN}}{\ev_{\sN}}$'s. This can be traced back to the fact that $\scal{\hv_{\sN}}{\ev_{\sN}}$ and $\scal{\hv'_{\sN}}{\ev_{\sN}}$ are positively correlated when $\hv_{\sN}$ and $\hv'_{\sN}$  are close to each other in Hamming distance. Actually these correlations have an effect on the tails of the largest Fourier coefficients as demonstrated in Figure \ref{fig:problem} which display longer tails corresponding to the largest Fourier coefficients in the case of a noise produced by $\scal{\hv_{\sN}}{\ev_{\sN}}$'s (called parity-checks in the figure) instead of Fourier coefficients produced by decoding a code with a BSC$(\frac{1-\varepsilon}{2})$ noise (called BSC in the figure). This phenomenon vanishes when $k$ gets larger as can be verified in Figure 
\ref{fig:problem} or on \url{https://github.com/tillich/RLPNdecoding/tree/master/verification_heuristic/histogram}. From our experiments
(see more details on \url{https://github.com/tillich/RLPNdecoding}) this phenomenon is not severe enough to prevent Algorithm \ref{alg:RLPN} from working but needs some adjustments about how larger $N$ has to be in terms of $\frac{1}{\delta^2}$. This experimental evidence leads us to conjecture 
\begin{conj}\label{con:main}
Algorithm \ref{alg:RLPN} is successful if we replace in Proposition \ref{prop:complexity1} the condition $N = \omega\left(n \left(\frac{\binom{n-s}{w}}{K_w^{n-s}(u)}\right)^2 \right)$ by the slightly stronger condition
$N = \omega\left(n^\alpha \left(\frac{\binom{n-s}{w}}{K_w^{n-s}(u)}\right)^2 \right)$ for a certain $\alpha \geq 1$. 
\end{conj}
If this conjecture is true, then obviously the asymptotic exponent of the complexity is unchanged if we replace Assumption 
\ref{ass:LPN} by Conjecture \ref{con:main}.
A semi-heuristic way to verify this conjecture could be to proceed as follows
\begin{enumerate}
\item Let $W$ be the weight of the vector $\left( \scal{\widetilde{\hv}_{\sN}}{\ev_{\sN}}\right)_{\hv \in \tilde{\sH}}$. Compute $\var{W}$ and prove that $\var{W}$ is of order $\OO{n^\beta N}$ where $\beta$ is some constant.
\item Use this computation to bound heuristically the tails of the Fourier coefficients and use this computation of  $\var{W}$  to give an estimation for the second largest Fourier coefficient when decoding the $[N,s]$-code which agrees with the experimental evidence. 
\item Use this to prove that the second largest Fourier coefficient is typically far away enough from the first one to prove the validity of Conjecture 
\ref{con:main}.
\end{enumerate}}
\begin{figure}[h!]
\begin{center}
\begin{tabular}{c}
\includegraphics[height=4cm]{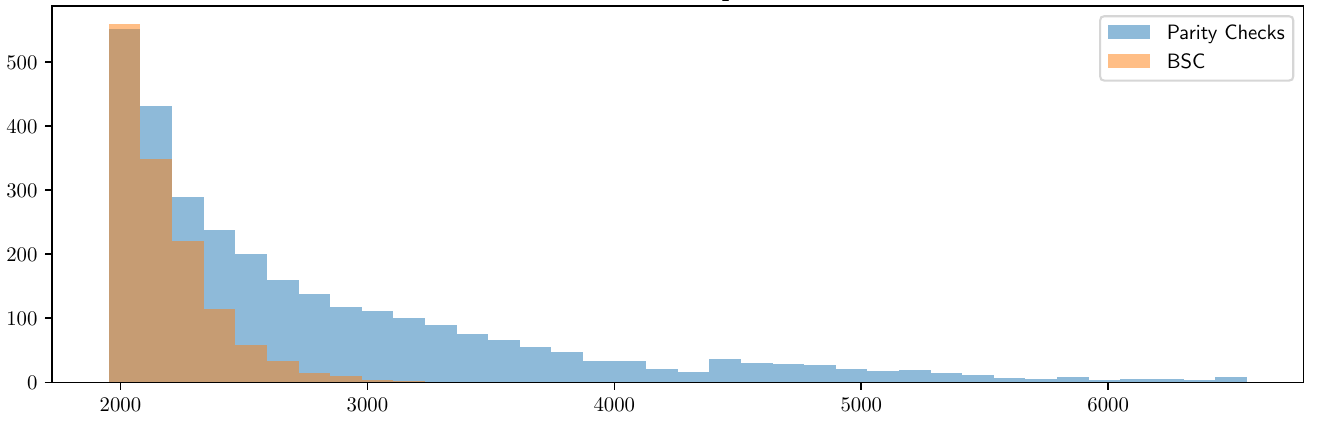} \\
\includegraphics[height=4cm]{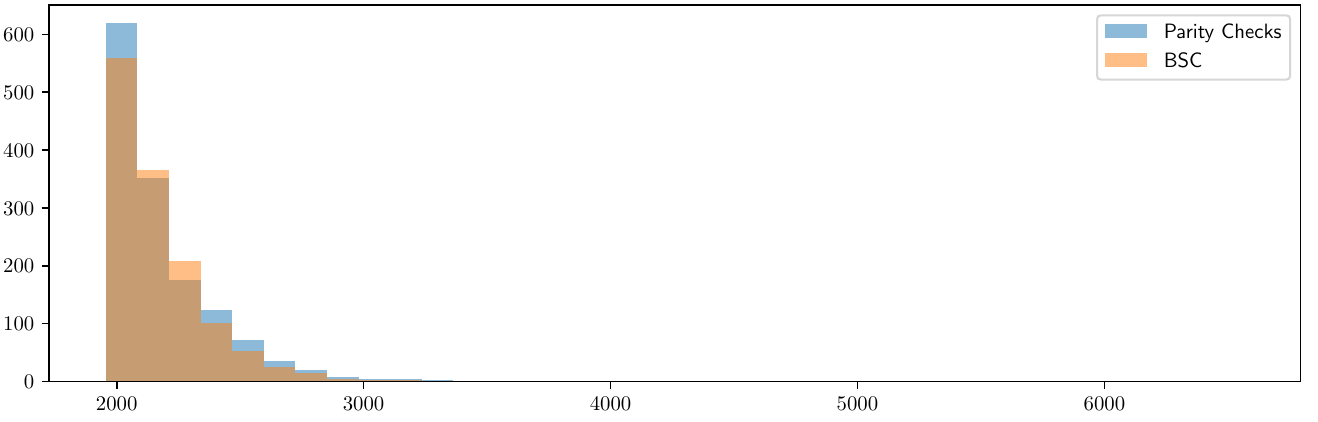}
\end{tabular}
\end{center}
\caption{\label{fig:problem} Tails of the largest Fourier coefficients when decoding the $[N,s]$-code either with the noise produced by the 
$\scal{\hv_{\sN}}{\ev_{\sN}}$'s or by the ideal LPN noise model (the BSC$(\frac{1-\varepsilon}{2})$ noise model). 
Both figures correspond to parity-checks $\hv_{\sN}$ of weight $6$ and to $s=19$. However they differ in the value for $k$.
$k$ equals $26$ in the first figure and displays rather heavy tails for the largest Fourier coefficients corresponding to the parity-checks $\hv_{\sN}$ whereas $k=40$ corresponds to rather similar tails
in both cases. This is a general trend that can be verified on {\url{https://github.com/tillich/RLPNdecoding/tree/master/verification_heuristic/histogram}}, when $k$ gets larger, the heavy tail phenomenon vanishes.}
\end{figure}

 \section{Collision techniques for finding low weight parity-checks}
\label{sec:Dumer}

\subsection{Using the \cite{D86} method}\label{subsecD86}
A way for creating  parity-checks with a low weight on $\sN$ is simply to use subset-sum/collision techniques \cite{D86,S88,D89}.
We start here with the simplest method for performing such a task pioneered by Dumer in \cite{D86}.
Consider a parity-check matrix $\Hm$ for the code $\CC$ we want to decode and keep only the columns 
belonging to $\sN$ to obtain an $(n-k)\times (n-s)$ matrix $\Hm_\sN$. The row-space of $\Hm_\sN$ generates the restrictions $\hv_\sN$ to $\sN$ of the parity-checks $\hv$ of $\CC$.
This row-space is nothing but the dual code $\CC^\perp$ punctured in $\sP$, \ie we keep only the positions in $\sN$. With our notation, this is $\CC^\perp_\sN$ and is an $[n-s,n-k]$-code. Therefore if we want to find parity-checks $\hv$ of $\CC$ such that $|\hv_\sN|=w$, this amounts to find codewords of $\CC^\perp_\sN$ of weight $w$. For this, we compute a parity-check matrix $\Hm'$ of $\CC^\perp_\sN$ \ie a $(k-s)\times (n-s)$ matrix such that 
$$
\CC^\perp_\sN = \{\cv \in \Ft^{n-s}: \Hm' \transp{\cv}=0\}.
$$
We split such a matrix in two parts\footnote{To simplify the presentation, the cut is explained by taking the first $\frac{n-s}{2}$ positions for the first part and the $\frac{n-s}{2}$ for the second part, but of course in general these positions are randomly chosen.} of the same size $\Hm'=\begin{pmatrix} \Hm_1 & \Hm_2 \end{pmatrix}$. We obtain an algorithm of time  and space complexity, $\Ctime$ and $\Cspace$ respectively, producing $N$ codewords of weight $w$, with  
\begin{equation*}
N  =  \frac{\binom{\frac{n-s}{2}}{\frac{w}{2}}^2}{2^{k-s}}(1+o(1)) \quad \mbox{and} \quad 
\Cspace = \Ctime = \OO{\binom{\frac{n-s}{2}}{\frac{w}{2}}+N}.
\end{equation*}
The algorithm for producing such codewords is to set up two lists, 
$$
\sL_1\eqdef \left\{(\Hm_1 \transp{\hv_1},\hv_1): |\hv_1|=\frac{w}{2},\; \hv_1 \in \Ft^{\frac{n-s}{2}}\right\}
$$ 
$$
\sL_2 \eqdef \left\{(\Hm_2 \transp{\hv_2},\hv_2): |\hv_2|=\frac{w}{2}, \; \hv_2 \in \Ft^{\frac{n-s}{2}}\right\}
$$ 
and looking for collisions $\Hm_1 \transp{\hv_1}=\Hm_2 \transp{\hv_2}$  in the lists. It yields vectors $\hv'=\hv_1 ||\hv_2$ of weight $w$ which are in $\CC^\perp_\sN$ since $\Hm' \transp{\hv'}=\Hm_1 \transp{\hv_1}+\Hm_2 \transp{\hv_2}=\mathbf{0}$. These vectors in $\Ft^{n-s}$ can be completed to give vectors $\hv \in \Ft^n$ such that $\hv_\sN=\hv'$. The number of collisions is expected to be of order ${\binom{\frac{n-s}{2}}{\frac{w}{2}}^2}/{2^{k-s}}$ since $2^{-(k-s)}$ is the collision probability of two vectors in $\Ft^{k-s}$. 
The algorithm for performing this task is given by Algorithm \ref{alg:collision}.

\begin{algorithm}[h!]
\caption{Creating low weight parity-checks by collisions \label{alg:collision}}
\hspace*{\algorithmicindent} \textbf{Input} $\CC$, $w$, $\sP$\\
 \hspace*{\algorithmicindent} \textbf{Output} a list of parity-check equations $\hv$ of $\CC$ such that $|\hv_\sN|=w$ where 
$\sN\eqdef\llbracket 1, n \rrbracket \setminus \sP$.
\begin{algorithmic}
\Function{Create}{$\CC$,$w$, $\sP$}
\State $\Hm \gets \Call{Parity-check-matrix}{\CC^\perp,\sP}$
\State
\Comment{returns a parity-check matrix for $\CC^\perp$ with an identity corresponding to the positions in $\sP$: \mbox{ }  \mbox{ } \mbox{ } \mbox{ } \mbox{ } \mbox{ } 
$\Hm = \begin{pmatrix} \Id &\Pm \\ \vec 0 & \Hm' \end{pmatrix}$ where we assume that the first block corresponds to the positions of $\sP$.}
\State $\sL_1 \gets \{(\Hm_1 \transp{\hv_1},\hv_1): |\hv_1|=w/2, \hv_1 \in \Ft^{\frac{n-s}{2}}\}$
\State $\sL_2 \gets \{(\Hm_2 \transp{\hv_2},\hv_2): |\hv_2|=w/2, \hv_2 \in \Ft^{\frac{n-s}{2}}\}$
\State
\Comment{We assume $\Hm'=\begin{pmatrix} \Hm_1 & \Hm_2 \end{pmatrix}$, with $\Hm_1$ and $\Hm_2$ of the same size.}
\State $\sL \gets \{\hv_1 ||\hv_2 \in \sL_1 \times \sL_2: \Hm_1\transp{\hv_1}=\Hm_2\transp{\hv_2}\}$
\State
\Return $\left\{ \hv' \transp{\Pm} ||\hv' : \hv' \in \sL \right\}$
\State
\Comment{It is straightforward to check that $\hv' \transp{\Pm} ||\hv'$ belongs to $\CC^\perp$.}
\EndFunction
\end{algorithmic}
\end{algorithm}	

We have represented in Figure \ref{fig:Dumer0} the form of the parity-checks output by this method, together with the bet we make on the error.
\begin{center}
	\begin{figure}[h!] 
		\begin{center}
		\scalebox{0.95}{ 
			\begin{tikzpicture}
				\node at (-0.5,0.25) {$\vec{h}$};
				\draw (0,0) rectangle (12,0.5);
				\filldraw[pattern = north east lines] (0,0) rectangle (2.5,0.5);
				\draw[<->, >=stealth] (0.025,0.75) -- (2.45,0.75);
				\node at (1.25,1) {$s$}; 
				\draw (2.5,0) rectangle (7.25,0.5);
				\draw[<->, >=stealth] (2.525,0.75) -- (7.2,0.75);
				\node at (4.875,1) {$(n-s)/2$}; 
				\node at (4.875,0.25) {\small $w/2$};
				\draw (7.25,0) rectangle (12,0.5);
				\draw[<->, >=stealth] (7.3,0.75) -- (11.95,0.75);
				\node at (9.625,1) {$(n-s)/2$}; 
				\node at (9.625,0.25) {$w/2$};

				\node at (-0.5,-1.25) {$\vec{e}$};
				\draw (0,-1.5) rectangle (12,-1);
				\draw (0,-1.5) rectangle (2.5,-1);
\node at (1.25,-1.25) {$t-u$}; 
				\draw (2.5,-1.5) rectangle (12,-1);
\node at (7.25,-1.25) {$u$};
			\end{tikzpicture}
		}
		\end{center}
	\caption{The form of the parity-checks produced by this method, \textit{vs.} the bet made on the error. The hatched rectangle of size $s$ for $\hv$ indicates that the weight is arbitrary on this part. \label{fig:Dumer0}}
	\end{figure}
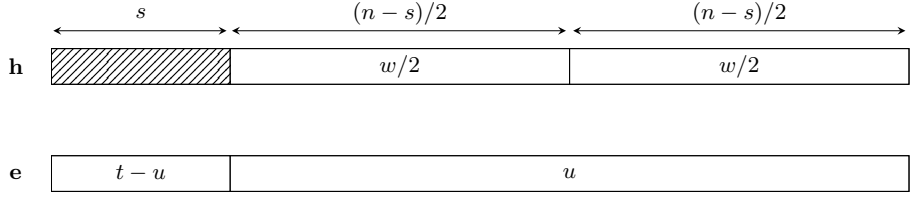
\end{center}

The amortized cost for producing a parity-check equation of weight $w$ is $O(1)$ as long as
$N \geq \Om{\binom{\frac{n-s}{2}}{\frac{w}{2}}}$. It is insightful to consider the smallest value of $w$ for which 
$\binom{\frac{n-s}{2}}{\frac{w}{2}} \leq {\binom{\frac{n-s}{2}}{\frac{w}{2}}^2}/{2^{k-s}}$. This is roughly speaking 
the smallest value (up to negligible terms) of $w$ for which the amortized cost for producing parity-check equations 
of weight $w$ is $O(1)$ per equation. In such a case, we roughly have
$$
N \approx \binom{\frac{n-s}{2}}{\frac{w}{2}}  \approx 
\frac{\binom{\frac{n-s}{2}}{\frac{w}{2}}^2}{2^{k-s}} \approx 2^{k-s}.
$$
In other words with this choice we have
$$
\Teq  =  \OO{2^{k-s}}.
$$
Let us choose now $u$ as the ``typical error weight'' when restricted to $\sN$, namely $u \approx t \frac{n-s}{n}$ and 
$s$ such that the decoding complexity of the $[N,s]$-code is also of order the codelength, \ie $N = \Tht{2^s}$. This would imply $2^s \approx 2^{k-s}$, which means that we are going to choose $s=\frac{k}{2}$. By using Proposition \ref{prop:complexity1}, all these choices 
would yield a time complexity $\CtimeD{86}$ for decoding $\CC$ which would be of order 
\begin{equation}
\label{eq:complexity_Dumer0}
\CtimeD{86} = \OOt{2^{k/2}},
\end{equation}
if the constraint $N = \Omt{\left(\frac{\binom{n-s}{w}}{K_w^{n-s}(u)}\right)^2}$ for successful decoding the $[N,s]$-code is met.
This amounts to 
$2^{Rn/2} = \Omt{\left(\frac{\binom{n(1-R/2)}{w}}{K_w^{n(1-R/2)}(t(1-R/2))}\right)^2}$,
where $R$ is the code rate, \ie $R=\frac{k}{n}$. By using Proposition \ref{prop:bias}, we can give an asymptotic formula for this constraint. It translates into
$$
R/2  \geq  2(1-R/2) \; \widetilde{\delta}\left(\tau,\omega/(1-R/2)\right),
$$
where $\widetilde{\delta}$ is the function defined in Proposition \ref{prop:bias0}. Amazingly enough this constraint is met up to very small values of $R$, it is only below $R \approx 0.02$ that this condition is not met anymore. This innocent looking remark has actually very concrete consequences. This means that above the range $R \gtrapprox 0.02$ the asymptotic complexity exponent, \ie $\alphaD{86} \eqdef \limsup_n \log_{2} \CtimeD{86} /n$ where 
$\CtimeD{86}$ is the time complexity, satisfies 
\begin{equation}
\label{eq:Dumer86}
\alphaD{86} \leq \frac{R}{2}.
\end{equation}
This is very surprising, since in the vicinity of $R \approx 0$ the asymptotic time complexity of {\em all known} decoding methods approach quickly $R$. In other words, in this regime, the complexity is of order $\Ctime \approx  2^{Rn}= 2^k$ for full distance ({\em a.k.a.} GV) decoding, meaning that they are not better than 
exhaustive search. Unfortunately this is also the case for our method. It can namely be proved that even by optimizing on the value of $s$, $w$ and $u$ we can not do better than this with our method, since $\alphaD{86}(R) \sim R$ as $R$ approaches $0$. However, 
as can be guessed from the fact that $\alpha_{\text{Dumer86}} \leq \frac{R}{2}$ for $R \gtrapprox 0.02$, the behaviour of the complexity is much better for our RLPN decoder. This can be verified in Figure \ref{fig:RLPNDumer86vsPrange}.

It is worthwhile to recall that ISD algorithms in the regime of the rate {\em close to $1$} precisely use this collision method to find low weight codewords in order to reduce significantly the complexity of decoding. 
In a sense, we have a dual version of the birthday/collision decoder of \cite{D86} with reduced complexity for rates close to $0$.

\begin{center} 
\begin{figure}[h!]
\centering
\scalebox{0.975}{
\begin{tabular}{cc}
\includegraphics[height=5cm]{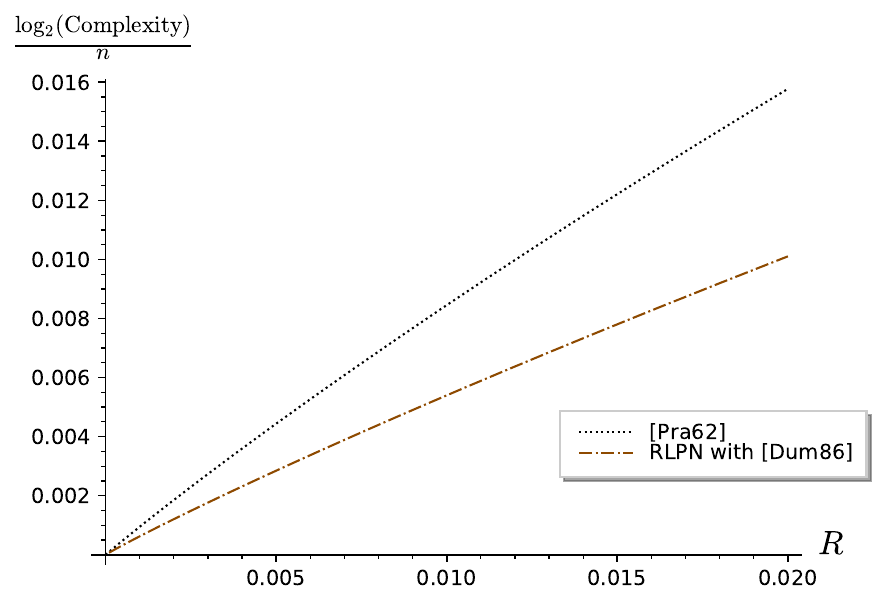} &\includegraphics[height=5cm]{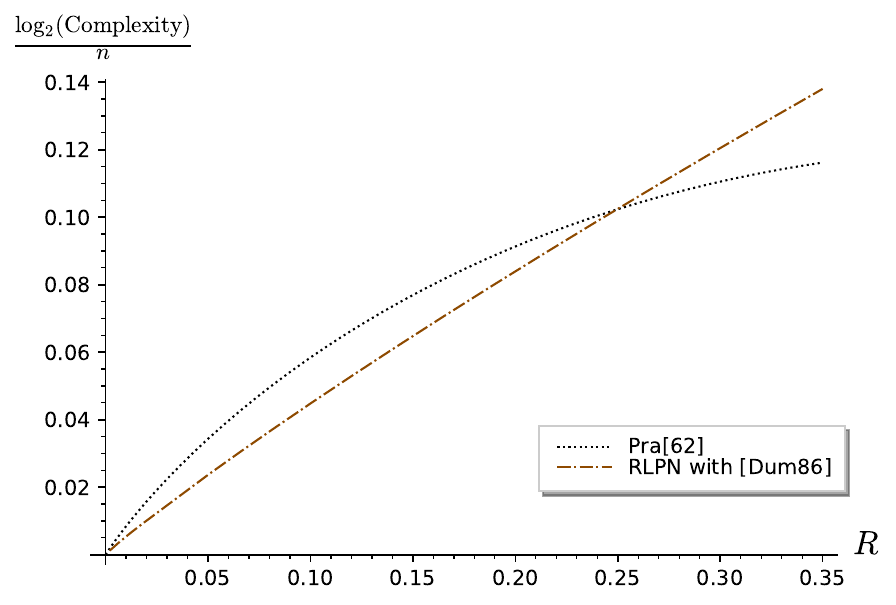}
\end{tabular}
}
\caption{The complexity of the RLPN-decoder for very small rates \emph{vs.} the simplest information set decoder, namely the ISD Prange decoder
	\cite{P62}. For small $R$, there is no much difference between the ISD Prange decoder and much more evolved decoders like \cite{BJMM12,MO15,BM17, BM18}. The RLPN-decoder with the very simple \cite{D86} technique performs much better for small rates than ISD decoders. It is only outperformed by the Prange decoder for rates above $0.25$ approximately. \label{fig:RLPNDumer86vsPrange} }
\end{figure}
\end{center}

\subsection{Improving \cite{D86}  by puncturing as in  \cite{D89}}\label{subsec:D91}
 There is a simple way of improving the generation of dual codewords of low weight on  $\sN$. It consists in partitioning $\sN$ in two sets $\sN_1$ and $\sN_2$ with 
$\sN_2$ being a subset of positions of size just a little bit above $n-k$ (which is the dimension of the dual code $\CC^\perp$), say $n-k+\ell$ and then to use the collision method to get dual codewords of weight $w_2$ on $\sN_2$. The same method is used in the improvement \cite{D89} of the simple collision decoder \cite{D86} or in a slightly less efficient way in \cite{S88}. It just consists in finding codewords in $\CC^\perp$ which have weight $w_1$ on $\sN_1$ and $w_2$ on $\sN_2$ instead of simply weight $w$ on $\sN$. We have represented in Figure \ref{fig:Dumer1} the form of the parity-checks we produce with this method. Note that the weight $w_1$ is expected to be half the size $k-\ell-s$ of $\sN_1$.
 
 \begin{center}
 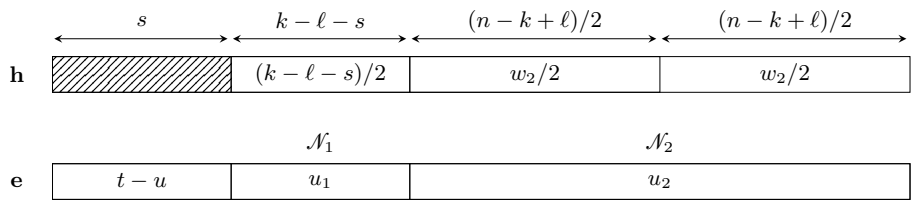
\begin{figure} 
 \begin{center}
 \scalebox{0.95}{ 
	\begin{tikzpicture}
	\node at (-0.5,0.25) {$\vec{h}$};
	\draw (0,0) rectangle (12,0.5);
	\filldraw[pattern = north east lines] (0,0) rectangle (2.5,0.5);
	\draw[<->, >=stealth] (0.025,0.75) -- (2.45,0.75);
	\node at (1.25,1) {$s$}; 
	\draw (2.5,0) rectangle (5,0.5);
	\draw[<->, >=stealth] (2.525,0.75) -- (4.95,0.75);
	\node at (3.75,1) {$k-\ell -s$}; 
	\node at (3.75,0.25) {\small $(k-\ell-s)/2$};
	\draw (5,0) rectangle (8.5,0.5);
	\draw[<->, >=stealth] (5.025,0.75) -- (8.45,0.75);
	\node at (6.75,1) {$(n-k+\ell)/2$}; 
	\node at (6.75,0.25) {$w_2/2$};
		\draw[<->, >=stealth] (8.525,0.75) -- (11.95,0.75);
	\node at (10.25,1) {$(n-k+\ell)/2$}; 
	\node at (10.25,0.25) {$w_2/2$};
	
	\node at (3.75,-0.7) {$\sN_1$};
		\node at (8.5,-0.7) {$\sN_2$};
	\node at (-0.5,-1.25) {$\vec{e}$};
	\draw (0,-1.5) rectangle (12,-1);
	\draw (0,-1.5) rectangle (2.5,-1);
\node at (1.25,-1.25) {$t-u$}; 
	\draw (2.5,-1.5) rectangle (5,-1);
\node at (3.75,-1.25) {$u_{1}$};
	\draw (5,-1.5) rectangle (12,-1);
\node at (8.5,-1.25) {$u_{2}$};
	\end{tikzpicture}
}
\end{center}
\caption{The form of the parity-checks produced by this method, \emph{vs.} the bet made on the error. The hatched rectangle of size $s$ for $\hv$ indicates that the weight is arbitrary on this part. \label{fig:Dumer1}}
\end{figure}

\end{center}

To understand the bias we get in this case, the proof of Proposition \ref{prop:bias} can be readily adapted to yield
	 \begin{proposition}\label{prop:bias2}
	 Assume that the code $\CC$ is chosen by picking for it an $(n-k) \times n$ binary parity-check matrix uniformly at random. 
	 Let $\sN$ be a fixed set of $n-s$ positions in $\IInt{1}{n}$ which is partitioned in two sets $\sN_1$ and $\sN_2$ and $\ev$ be some error of weight $u_i$ on $\sN_i$ for $i \in \{1,2\}$. For $i \in \{1,2\}$, choose $\hv$ uniformly at random among the parity-checks of $\CC$ of weight $w_i$ on the $\sN_i$'s and $\hv'$ uniformly at random among the words of weight $w_i$ on the $\sN_i$'s. For $i \in \{1,2\}$, let 
\begin{eqnarray*}
\delta_i & \eqdef & \bias\left( \scal{\ev_{\sN_i}}{\hv'_{\sN_i}} \right)\\
\delta & \eqdef & \delta_1 \delta_2
\end{eqnarray*} If the parameters $k$, $s$, $u_i$, $w_i$ are chosen as functions on $n$ so that  for $n$ going to infinity, the expected number $N$ of parity-checks of $\CC$ of respective weight $w_i$ on $\sN_i$ for $i \in \{1,2\}$, satisfies 
	 $N = \omega\left( 1 / \delta^2 \right)$ then for all but a proportion $o(1)$ of codes we have
	 $$
	 \bias\left( \langle \vec{e}_{\sN},  \vec{h}_{\sN}\rangle \right) = (1+o(1)) \delta.
	 $$
	 \end{proposition}

  With the collision method we use,  the parity-checks we produce have actually a slightly more specific form, since $\sN_2$ is partitioned in two sets of (almost) the same size on which $\hv$ has weight $w_2/2$. It is not difficult to turn such a generation of parity-checks at the cost of a polynomial overhead into a generation of uniformly distributed parity-checks of weight $w_2$ on $\sN_2$. We leave out the details for doing this here. Under such an assumption, we have
	\begin{lemma}\label{lem:bias2} 
		With the same assumptions as in Proposition \ref{prop:bias2}, 
		\begin{equation*}
		\mathbb{P}_{\vec{h}}(\langle \vec{e}_{\sN},  \vec{h}_{\sN}=1\rangle) = \frac{1 - \varepsilon}{2} \quad \mbox{where} \quad \varepsilon = \delta_1  \delta_2 (1-o(1)) 
\end{equation*}
		with $\delta_1 \eqdef  \frac{ K_{w_1}^{k-\ell-s}(u_1)}{\binom{k-\ell-s}{w_1}}$, $\delta_2  \eqdef  \frac{ K_{w_2}^{n-k+\ell}(u_2)}{\binom{n-k+\ell}{w_2}}$,  $u_1 \eqdef |\ev_{\sN_1}|$, $u_2 \eqdef |\ev_{\sN_2}|$, $w_1 \eqdef |\vec{h}_{\sN_1}|$ 
		and  $w_2 \eqdef |\vec{h}_{\sN_2}|$. 	
		\end{lemma}
		
\begin{proof}
This is a straightforward application of the previous proposition and 
 Lemma \ref{lem:bias}.
\end{proof}

All these considerations lead to a slight variation of the RLPN decoder given in Algorithm \ref{alg:RLPN}. Let us make now a bet on the weight 
$u_i$ of the error restricted to $\sN_i$ for $i \in \{1,2\}$ and use Dumer's \cite{D89} collision low-weight codeword generator to produce $N$ parity-checks $\hv$ such that $\left|\hv_{\sN_i}\right|=w_i$ for $i \in \{1,2\}$. We call the associated function \Call{Create}{$N$,$w_1$,$w_2$, $\sP$}. 

\begin{proposition}\label{prop:complexity3}
If Assumption \ref{ass:LPN} holds and 
assuming that \Call{Create}{$N$,$w_1$,$w_2$, $\sP$} produces $N$ parity-check equations in space $\Seq$ and time 
$\Teq$ that are of weight $w_i$ on $\sN_i$ for $i \in \{1,2\}$. The probability $\Psucc$ (over the choice of 	$\sN_1$ and $\sN_2$) that there are exactly $u_1$ errors in $\sN_1$ and $u_2$ errors in $\sN_2$ is given by
$$
\Psucc = \frac{\binom{s}{t-u_1-u_2}\binom{k-\ell-s}{u_1}\binom{n-k+\ell}{u_2}}{\binom{n}{t}}.
$$ The space complexity $\CspaceD{89}$ and time complexity $\CtimeD{89}$ of the RLPN-decoder are given by 
$$
\textbf{\textup{Space:}}\;\; \CspaceD{89}=\OO{\Seq+2^s}, \quad \textbf{\textup{Time:}}\;\; \CtimeD{89} = \OOt{\frac{\Teq+2^s}{\Psucc}}.
$$
under the constraint on the parameters $s$, $\ell$, $u_1$, $u_2$, $w_1$  and $w_2$ given by
\begin{eqnarray}
N & \leq & 2^s \label{eq:constraint_N}\\
N & \leq & \frac{\binom{k-\ell-s}{w_1}\binom{n-k+\ell}{w_2}}{2^{k-s}} \\
N &= &\om{\left(\frac{\binom{k-\ell-s}{w_1}\binom{n-k+\ell}{w_2}}{K_{w_1}^{k-\ell-s}(u_1)K_{w_2}^{k-\ell-s}(u_2)}\right)^2}.
\end{eqnarray}
\end{proposition}
We have found out that choosing $w_1$ carefully is unnecessary and simply setting it to it its expected value is sufficient, \ie
$w_1 = \frac{k-\ell-s}{2}$. Again, the same discussion as in the previous section applies and if Conjecture \ref{con:main} applies then 
the asymptotic form of the complexity is the same as if we use Proposition \ref{prop:complexity3} and we get  the following asymptotic form
\begin{proposition}\label{pro:RLPNDumer}If Conjecture \ref{con:main} holds, the asymptotic complexity exponent of the RLPN decoder based on Dumer's collision low weight dual codeword generators is given by
\begin{equation}
\alpha_{\text{Dumer89}}(R) \eqdef \min_{(\sigma,\nu_1,\nu_2,\lambda,\omega_1,\omega_2) \in \sR} \beta(R,\sigma,\nu_1,\nu_2,\lambda,\omega_1,\omega_2)
\end{equation}
where
$$
\beta \eqdef \max\left( \sigma,\nu' \right)+ \pi,
$$
$$
\nu' \eqdef \max{ \left( \frac{(1-R+\lambda)}{2} h\left( \frac{\omega_2}{1-R+\lambda}\right) , \nu \right)}, \quad \nu \eqdef (1-R+\lambda) h\left( \frac{\omega_2}{1-R+\lambda}\right) -\lambda, \\
$$
$$
\pi \eqdef 1-R -\sigma h\left( \frac{\tau-\nu_1-\nu_2}{\sigma}\right) 
-(R-\lambda-\sigma) h\left( \frac{\nu_1}{R-\lambda-\sigma} \right) -(1-R+\lambda) h \left( \frac{\nu_2}{1-R+\lambda}\right),
$$
$$
\tau \eqdef \deltaGV(R) = h^{-1}(1-R)
$$
and the constraint region $\sR$ is defined by the subregion of non-negative tuples $(\sigma,\nu_1,\nu_2,\lambda,\omega_1,\omega_2)$ such that $\omega_1 =  \frac{R - \lambda - \sigma}{2}$
$$
\sigma \leq R - \lambda, \quad  \nu_1 \leq R- \lambda - \sigma, \quad   \nu_2 \leq 1- R + \lambda, \quad    \tau- \sigma \leq \nu_1+\nu_2 \leq \tau, \quad \nu \leq \sigma, 
$$
and
$$
\nu = - (R-\lambda-\sigma)\widetilde{\delta}(\frac{\nu_1}{R-\lambda-\sigma},\frac{\omega_1}{R-\lambda-\sigma}) - (1-R+\lambda)\widetilde{\delta}(\frac{\nu_2}{1-R+\lambda},\frac{\omega_2}{1-R+\lambda})
$$
where $\widetilde{\delta}$ is the function defined in Proposition \ref{prop:bias0}.

\end{proposition}
 \section{Using advanced collision techniques}\label{sec:BJMM}
ISD techniques have evolved  \cite{S88,D89,BLP11,MMT11,BJMM12} by first introducing \cite{S88} collision/subset-sum techniques  
whose purpose is to produce for codes of rate close to $1$, all codewords of some small weight, and later on by substantially improving  them by using on top of that for instance representation techniques \cite{MMT11}. These algorithms come very handy in our case for devising the function \Call{Create}{$N,w,\sP$} that we need. In the previous section, we have explored what could be achieved by the very first techniques of this type taken from \cite{D86,D89}. We are 
going to explain now what can be gained by using \cite{MMT11,BJMM12}. 
It is convenient here to formalize the basic step used in the previous section which can be explained by the following function
\begin{algorithm}
\hspace*{\algorithmicindent} \textbf{Input:} $\sL_1 \subseteq \Ft^n$, $\sL_2\subseteq \Ft^n$, $w \in \llbracket 1, n \rrbracket$, $\Hm \in \Ft^{\ell\times n}$\\
\hspace*{\algorithmicindent} \textbf{Output:} a list $\sL = \{\xv=\xv_1+\xv_2: \xv_i \in  \sL_i,\;i\in \{1,2\},\;|\xv|=w,\;\Hm \transp{\xv} = \mathbf{0} \}$ of elements
of the form $\xv_1+\xv_2$ with $\xv_i$ belonging to $\sL_i$ of weight $w$ belonging to the code of parity-check matrix $\Hm$
\begin{algorithmic}
\Function{Merge}{$\sL_1,\sL_2,w,\Hm$}
\State
$\sL \gets \emptyset$
\ForAll{$\xv_1 \in \sL_1$}
\State
Store $\xv_1$ in a hashtable $\sT$ at address $\Hm \transp{\xv_1}$
\EndFor
\ForAll{$\xv_2 \in \sL_2$}
\If{$\exists\xv_1$ in $\sT$ at address $\Hm \transp{\xv_2}$ and $|\xv_1+\xv_2|=w$}
\State
Put $\xv_1+\xv_2$ in $\sL$
\EndIf
\EndFor
\State \Return $\sL$
\EndFunction
\end{algorithmic}
\end{algorithm}

It creates codewords of weight $w$ in a code of parity-check matrix $\Hm$ as sums $\xv_1 + \xv_2$ of two lists $\sL_1$ and $\sL_2$ with a complexity which is of the form
$\OO{\max\left(\#\sL_1,\#\sL_2,\frac{\#\sL_1\cdot\#\sL_2}{2^\ell} \right)}$ if the $\Hm \transp{\xv_i}$'s are distributed uniformly at random and independently (we will make this assumption from now on). It is clear that \cite{D86} and \cite{D89} is more or less a direct application of this method. \cite{MMT11} and \cite{BJMM12} use several layers of this function.
\cite{MMT11} starts by partitioning the set of positions of the vectors of $\Ft^n$ which are considered in two sets $\sI_1$ and $\sI_2$ of about the same size. Then it starts with two lists $\sL_1^0$ and $\sL^0_2$ of all elements of weight $p_0$ and support $\sI_1$ and $\sI_2$ respectively. It merges them in a list $\sL^1$ of elements of weight $p_1$ in the kernel of a parity-check matrix $\Hm_1$. Since the elements of $\sL_1^0$ and $\sL^0_2$ have disjoint supports by construction, we necessarily have that $p_1=2p_0$. List $\sL^1$ is then merged with itself to yield elements which are in the kernel of another matrix $\Hm_2$ (see Figure \ref{fig:MMTBJMM}). Since these are sums of elements of $\sL^1$ they are also in the kernel of $\Hm_1$, so that that the elements of the final list are of weight $p_2$ and belong to the code of parity-check 
$\Hm = \begin{pmatrix} \Hm_1 \\ \Hm_2 \end{pmatrix}$. The size of $\Hm_1$ is chosen such that an element $\xv$ of weight $p_2$ and $\Hm_1 \transp{\xv}=0$ is typically the sum of only two elements of $\sL^1$ (this is the point of the representation technique). \cite{BJMM12} is similar to \cite{MMT11} with one layer which is added. In this case, we create at the end a list of elements of weight $p_3$ which are in the code of parity-check matrix 
\begin{equation}\label{eq:H}
\Hm \eqdef \begin{pmatrix} \Hm_1 \\ \Hm_2 \\ \Hm_3  \end{pmatrix}.
\end{equation} 
The sizes of $\Hm_1$ in the \cite{MMT11} case, and of $\Hm_1$ and $\Hm_2$ in \cite{BJMM12} are chosen to ensure unicity of the representation of an element of a list as the sum of two elements of the lists used for the merge (this is the representation technique). 
\begin{center}
	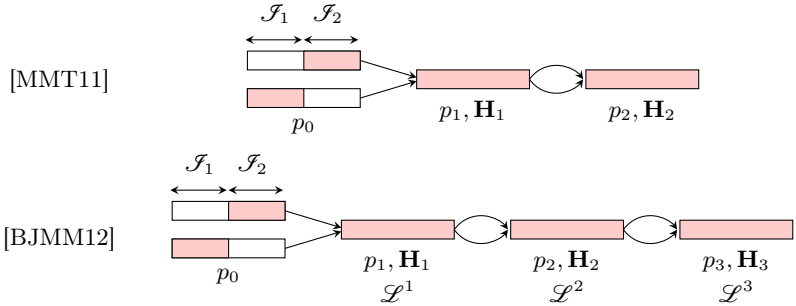
\begin{figure}[h!]
	\begin{center}
		\begin{tikzpicture}
			\node at (0,0.1) {\cite{MMT11}};
			\draw (2.5,0.25) rectangle (4,0.5);
			\filldraw[fill = red!20] (3.25,0.25) rectangle (4,0.5);
			\node at (2.875,1) {$\sI_{1}$};
			\draw[<->, >=stealth] (2.475,0.65) -- (3.225,0.65);
			\node at (3.575,1) {$\sI_{2}$};
			\draw[<->, >=stealth] (3.25,0.65) -- (3.975,0.65);
			\draw[->, >=stealth] (4,0.375) -- (4.75,0.15);
			
			\draw (2.5,-0.25) rectangle (4,0);
			\filldraw[fill = red!20] (2.5,-0.25) rectangle (3.25,0);
			\node at (3.25,-0.5) {$p_{0}$};
			\draw[->, >=stealth] (4,-0.125) -- (4.75,0.1);
				
			\filldraw[fill = red!20] (4.75,0) rectangle (6.25,0.25); 
			\node at (5.5,-0.3) {$p_{1},\vec{H}_{1}$};
			\filldraw[fill = red!20] (7,0) rectangle (8.5,0.25); 
			\node at (7.75,-0.3) {$p_{2},\vec{H}_{2}$};
			\draw[->, >=stealth] (6.25,0.125) arc (145:35:0.435cm);
			
			\draw[->, >=stealth] (6.25,0.125) arc (-145:-35:0.435cm);

			\node at (0,-2) {\cite{BJMM12}};
			\draw (1.5,-1.75) rectangle (3,-1.5);
			\filldraw[fill = red!20] (2.25,-1.75) rectangle (3,-1.5);
			\node at (1.875,-1) {$\sI_{1}$};
			\draw[<->, >=stealth] (1.475,-1.35) -- (2.225,-1.35);
			\node at (2.575,-1) {$\sI_{2}$};
			\draw[<->, >=stealth] (2.25,-1.35) -- (2.975,-1.35);
			\draw[->, >=stealth] (3,-1.625) -- (3.75,-1.85);
			
			\draw (1.5,-2.25) rectangle (3,-2);
			\filldraw[fill = red!20] (1.5,-2.25) rectangle (2.25,-2);
			\node at (2.25,-2.5) {$p_{0}$};
			\draw[->, >=stealth] (3,-2.125) -- (3.75,-1.9);
			
			\filldraw[fill = red!20] (3.75,-2) rectangle (5.25,-1.75); 
			\node at (4.5,-2.3) {$p_{1},\vec{H}_{1}$};
			\node at (4.5,-2.7) {$\sL^{1}$};
			\filldraw[fill = red!20] (6,-2) rectangle (7.5,-1.75); 
			\node at (6.75,-2.3) {$p_{2},\vec{H}_{2}$};
			\node at (6.75,-2.7) {$\sL^{2}$};
			\draw[->, >=stealth] (5.25,-1.875) arc (145:35:0.435cm);
			\draw[->, >=stealth] (5.25,-1.875) arc (-145:-35:0.435cm);
			
				\filldraw[fill = red!20] (8.25,-2) rectangle (9.75,-1.75); 
			\node at (9,-2.3) {$p_{3},\vec{H}_{3}$};
			\node at (9,-2.7) {$\sL^{3}$};
			\draw[->, >=stealth] (7.5,-1.875) arc (145:35:0.435cm);
			\draw[->, >=stealth] (7.5,-1.875) arc (-145:-35:0.435cm);
			
		\end{tikzpicture}
		\end{center}
	\caption{This figure represents the successive lists obtained in \cite{MMT11} and \cite{BJMM12}. The support of the elements of the list are represented in pink. Arrows point from the lists which are merged to the result of the merge and if two arrows depart from a list and arrive at another list, this means that the departure list is merged with itself. The weights of the elements are indicated for each level and the matrix $\Hm_i$ used for the merge is also given at the level of the result of the merge. \label{fig:MMTBJMM}}
	\end{figure}
\end{center}

We use these two techniques as we used the \cite{D86} technique inside the \cite{D89} technique, namely to generate codewords of $\CC^\perp$ (\ie $\Hm \transp{\xv}=\mathbf{0}$ for $\Hm$ given by \eqref{eq:H}) which are of weight $p_{3}$ on a set of indices of size $n-k+\ell$ (see Figure \ref{fig:BJMM}).

 \begin{center}
	\begin{figure} 
		\begin{center}
		\scalebox{0.95}{ 
			\begin{tikzpicture}
				\node at (-0.5,0.25) {$\vec{h}$};
				\draw (0,0) rectangle (12,0.5);
				\filldraw[pattern = north east lines] (0,0) rectangle (2.5,0.5);
				\draw[<->, >=stealth] (0.025,0.75) -- (2.45,0.75);
				\node at (1.25,1) {$s$}; 
				\draw (2.5,0) rectangle (5,0.5);
				\draw[<->, >=stealth] (2.525,0.75) -- (4.95,0.75);
				\node at (3.75,1) {$k-\ell -s$}; 
				\node at (3.75,0.25) {\small $(k-\ell-s)/2$
};
\draw[<->, >=stealth] (5.025,0.75) -- (11.95,0.75);
				\node at (8.5,1) {$n-k+\ell$}; 
				\node at (8.5,0.25) {$p_{3}$};

				\node at (-0.5,-1.25) {$\vec{e}$};
				\draw (0,-1.5) rectangle (12,-1);
				\draw (0,-1.5) rectangle (2.5,-1);
\node at (1.25,-1.25) {$t-u$}; 
				\draw (2.5,-1.5) rectangle (5,-1);
\node at (3.75,-1.25) {$u_{1}$};
				\draw (5,-1.5) rectangle (12,-1);
\node at (8.5,-1.25) {$u_{2}$};
			\end{tikzpicture}
		}
		\end{center}
		\caption{The form of the parity-checks produced by this method \cite{BJMM12}, \emph{vs.} the bet made on the error. The hatched rectangle of size $s$ for $\hv$ indicates that the weight is arbitrary on this part. \label{fig:BJMM}}
	\end{figure}
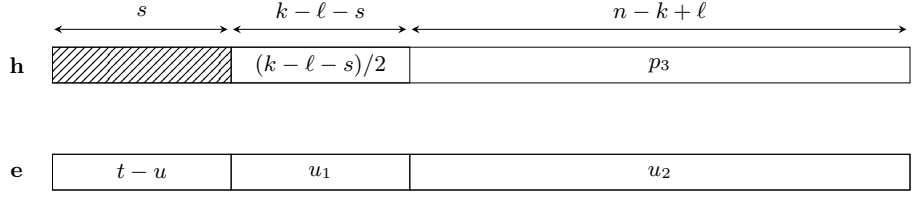
\end{center}

If we let $\ell_1$ be the number of rows of $\Hm_1$, $\ell_2$ be the number of rows of the matrix of $\Hm_2' \eqdef \begin{pmatrix} \Hm_1 \\ \Hm_2 \end{pmatrix}$, then the fact that the elements of $\sL^2$ should have a unique representation in terms of a sum of a pair of elements of $\sL^1$ respectively and that they are all elements $\xv$ of weight $p_1$ and $p_2$ respectively which satisfy $\Hm_2' \transp{\xv}=\mathbf{0}$ and $\Hm \transp{\xv}=\mathbf{0}$ respectively, imposes conditions \eqref{eq:representation} which follow. 
The $S_i$ represent the space complexity of the successive lists (\ie $\sL^0$, $\sL^1$, $\sL^2$ and $\sL^3$) used in the \cite{BJMM12} algorithm, whereas the $T_i$'s denote the complexity of each merge and $\Teq$ is the final complexity.
 \begin{align}
2^{\ell_1}  &=  \binom{p_2}{p_2/2}\binom{n-k+\ell-p_2}{p_1-p_2/2},
& 2^{\ell_2} &= \binom{p_3}{p_3/2}\binom{n-k+\ell-p_3}{p_2-p_3/2}  \label{eq:representation} 
\end{align}
\begin{equation}
S_0 =  \binom{\frac{n-k+\ell}{2}}{\frac{p_1}{2}}, \quad
S_1 =  \frac{\binom{n-k+\ell}{p_1}}{2^{\ell_1}}, \quad
S_2  =  \frac{\binom{n-k+\ell}{p_2}}{2^{\ell_2}}, \quad 
S_3  =  \frac{\binom{n-k+\ell}{p_3}}{2^{\ell}}
\end{equation}
\begin{equation}
T_0 = S_0, \quad
T_1  =  S_0 + \frac{S_0^2}{2^{\ell_1}}, \quad 
T_2  =  S_1 + \frac{S_1^2}{2^{\ell_2-\ell_1}}, \quad 
T_3  =  S_2 + \frac{S_2^2}{2^{\ell-\ell_2}},\quad
\Teq  =  T_0+T_1+T_2+T_3
\end{equation}

There is a similar proposition as Proposition \ref{pro:RLPNDumer} which gives the asymptotic complexity of the RLPN decoder used in conjunction with the \cite{MMT11} or \cite{BJMM12} techniques for producing low weight codewords. For \cite{MMT11} it is given by
\begin{restatable}{proposition}{MMT}\label{pro:RLPNMMT}
If conjecture \ref{con:main} applies, the asymptotic complexity exponent of the RLPN decoder based on \cite{MMT11} is given by
\begin{equation}
	\alpha_{\text{MMT}}(R) \eqdef \min_{(\sigma,\nu_1,\nu_2,\lambda,\lambda_1,\omega_1,\omega_2,\pi_1) \in \sR} \beta(R,\sigma,\nu_1,\nu_2,\lambda,\lambda_1,\omega_1,\omega_2,\pi_1)
\end{equation}
where
\begin{eqnarray*}
	\beta & \eqdef & \max(\sigma,\nu') + \pi, \\
	\nu' & \eqdef & \max(\gamma_1,\gamma_2), \quad \nu \eqdef (1-R+\lambda)h\left(\frac{\omega_2}{1-R+\lambda}\right)-\lambda, \\
	\gamma_1 & \eqdef & \max\left(\frac{1-R+\lambda}{2}h\left(\frac{\pi_1}{1-R+\lambda} \right),(1-R+\lambda)h\left(\frac{\pi_1}{1-R+\lambda} \right)-\lambda_1\right),\\
	\gamma_2 & \eqdef & 2(1-R+\lambda)h\left(\frac{\pi_1}{1-R+\lambda} \right)-\lambda_1-\lambda,\\
	\rho & \eqdef & 1-R -\sigma h\left( \frac{\tau-\nu_1-\nu_2}{\sigma}\right) -(R-\lambda-\sigma) h\left( \frac{\nu_1}{R-\lambda-\sigma} \right) -(1-R+\lambda) h \left( \frac{\nu_2}{1-R+\lambda}\right),\\
	\tau & \eqdef & \deltaGV(R) = h^{-1}(1-R)
\end{eqnarray*}
and the constraint region $\sR$ is defined by the subregion of non-negative tuples $(\sigma,\nu_1,\nu_2,\lambda,\lambda_1,\pi,\omega_1,\omega_2)$ such that
$$
\sigma \leq R-\lambda, \quad  \lambda_1 \leq \lambda, \quad  \pi_1 \leq \omega_2, \quad \nu_1 \leq R- \lambda - \sigma, \quad \nu_2 \leq 1- R + \lambda, \quad \nu \leq \sigma,
$$
$$
\tau- \sigma \leq \nu_1+\nu_2 \leq \tau, \quad \omega_1 = \frac{R - \lambda - \sigma}{2}, \quad \omega_2 < 1-R+\lambda,
$$
\begin{equation*}
\frac{\omega_2 }{2}< \pi_1 < 1-R+\lambda, \quad
\lambda_1 =  \omega_2 + (1-R+\lambda - \omega_2)h\left( \frac{\pi_1-\omega_2/2}{1- R+\lambda-\omega_2}\right) ,
\end{equation*}

$$
\nu= -\!(R-\lambda-\sigma)\widetilde{\delta}(\frac{\nu_1}{R-\lambda-\sigma},\frac{\omega_1}{R-\lambda-\sigma})
	-
	(1-R+\lambda)\widetilde{\delta}(\frac{\nu_2}{1-R+\lambda},\frac{\omega_2}{1-R+\lambda}) 
$$
where $\widetilde{\delta}$ is the function defined in Proposition \ref{prop:bias0}.

\end{restatable}
We give a proof  in the Appendix \ref{sec:RLPN-MMT-BJMM} for reference. There is a similar proposition for the asymptotic behaviour of RLPN decoding used together with \cite{BJMM12} which is given in the Appendix \ref{sec:RLPN-MMT-BJMM}. We have used them for producing the complexity curves given in Figure \ref{fig:lowerBound} which display the various complexities of the RLPN decoders we have presented. Even if there is a tiny improvement by using \cite{BJMM12} instead of \cite{MMT11} the two curves are nearly indistinguishable. A perspective of improvement of our algorithm could be to produce the parity-check equations by using more recent ISD techniques than \cite{BJMM12}, in particular \cite{MO15, BM17} or \cite{BM18} which all use nearest-neighbor search. Our preliminary results using in particular \cite{MO15} do not provide significant improvement, we have only been able to achieve a very slightly better complexity for rates close to $0.2$.

 \section{A Lower bound on the complexity of RLPN decoders}\label{sec:bound}

As pointed out all along the paper, RLPN decoding needs a large number $N$ of parity-check equations to work {\em but} of some shape as indicated below
\vspace*{-0.25cm}
\begin{center} 
		\scalebox{0.95}{ 
	\begin{tikzpicture}
		\node at (-0.5,0.25) {$\vec{h}$};
		\draw (0,0) rectangle (12,0.5);
		\filldraw[pattern = north east lines] (0,0) rectangle (2.5,0.5);
		\draw[<->, >=stealth] (0.025,0.75) -- (2.45,0.75);
		\node at (1.25,1) {$s$}; 
		\draw (2.5,0) rectangle (12,0.5);
\draw[<->, >=stealth] ((2.525,0.75) -- (11.95,0.75);
		\node at (7.25,1) {$n-s$}; 
		\node at (7.25,0.25) {$w$};

\end{tikzpicture}
}
\end{center} 
where the hatched area indicates that the weight is arbitrary on this part while $\vec{h}$ restricted on the other positions needs to have Hamming weight $w$. The number $N$ of such parity-checks has to verify (see Proposition \ref{prop:complexity1})
\begin{equation}\label{eq:cdtN0}  
N = \omega\left(n\left(\frac{\binom{n-s}{w}}{K_w^{n-s}(u)}\right)^2\right)
\end{equation} 
in order to be able to solve the underlying LPN problem. It can be verified that the smaller $w$ is (the bigger is the bias $\varepsilon$), the smaller is $N$ and the more efficient is our algorithm. Obviously if $w$ is too small, there are not enough such parity-checks.  
It can be verified that the expected number of parity-checks of the aforementioned shape is given by $2^{s} \binom{n-s}{w}/2^{k}$ in a random code (which is our assumption). Therefore we need \begin{equation}\label{eq:cdtN} 
N =\OO{ \frac{2^{s} \binom{n-s}{w}}{2^{k}}}.
\end{equation}

Given this picture it is readily seen that the complexity of RLPN decoding is always lower-bounded by $N$ (which is at least the cost to produce $N$ parity-checks) but we can be more accurate on the lower-bound over the complexity. Recall that we first need to solve an underlying LPN problem and that we make a bet on the number of errors $u$ in $\sN$. Therefore, {\em assuming} that we can compute a parity-check of the aforementioned shape in time $O(1)$, the complexity of this genie-aided RLPN decoding is given by 
\begin{equation}\label{eq:lowerbound}
\widetilde{O}\left( \frac{1}{\Psucc} \;\max \left( 2^{s}, N \right) \right) 
\end{equation}
where $\Psucc$ is given in Proposition \ref{prop:complexity1}. Our only constraints are given by \eqref{eq:cdtN0} and \eqref{eq:cdtN}. By optimizing \eqref{eq:lowerbound} over $s,u$ and $w$, we can give a lower-bound on the complexity of RLPN decoding. However notice that our lower-bound applies to a partition of parity-checks in two parts ($s$ and $n-s$).  We do not consider here finer partitions. This method for lower bounding the complexity of RLPN decoding is very similar to the technique used in \cite[\S 7]{DT17} to lower bound the complexity of statistical decoding. All in all, we give in Figure \ref{fig:lowerBound} this lower-bound of the complexity. The optimal parameters computed for each RLPN algorithms can be found on \url{https://github.com/tillich/RLPNdecoding}. As we see our RLPN decoders meet this lower-bound for small rates and we can hope to outperform significantly ISD's for code rates smaller than $\approx 0.45$. 
\begin{figure}[h!]
	\centering
	\includegraphics[height=6.2cm]{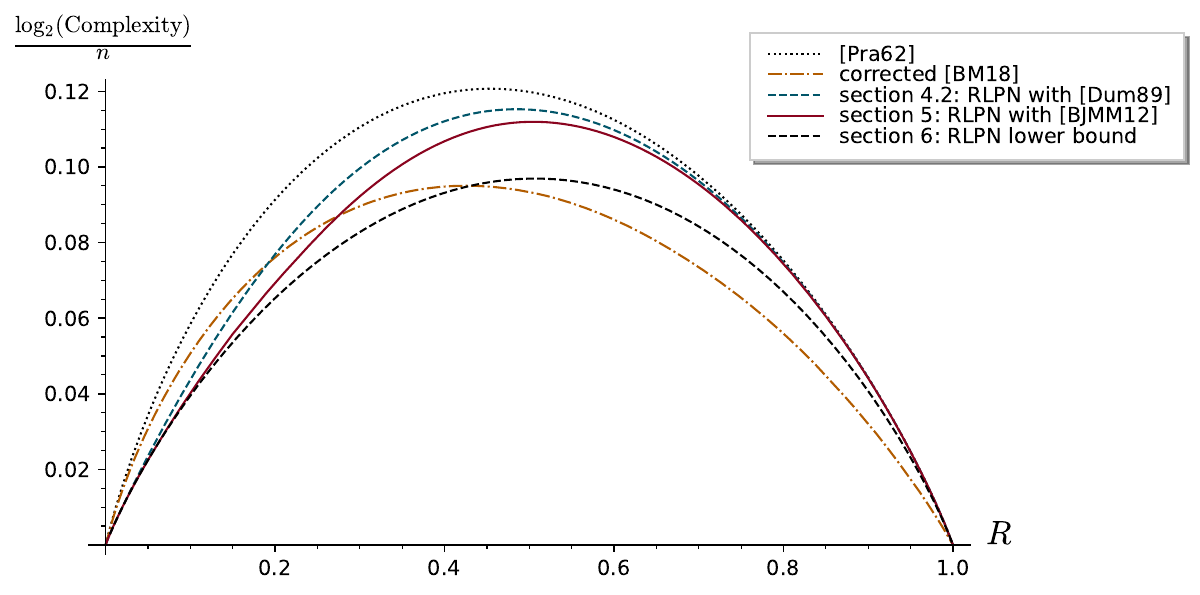}

	\caption{\label{fig:lowerBound} Complexity exponents of our different RLPN decoders, ISD's and the genie-aided RLPN algorithm when splitting parity-checks in two parts.}
\end{figure}

 \section{Concluding remarks}\label{sec:conclusion}

Since Prange's seminal work \cite{P62} in 1962, ISD algorithms have played a predominant role for assessing the complexity of 
code-based cryptographic primitives. In the fixed rate regime, they have been beaten only once in \cite{D86} with the help of collision techniques, and this only for a tiny code rate range ($R \in (0.98,1)$) and for a short period of time \cite{S88,D89} until these collision techniques were merged with the collision techniques to yield modern ISD's. Surprisingly enough, these improved ISD have resulted in decoding complexity curves tilting more and more to the left (\ie with a maximum which is attained more and more below $\frac{1}{2}$) instead of being symmetric around $\frac{1}{2}$ as it could have been expected. It is precisely for rates below $\frac{1}{2}$ that RLPN decoding is able to outperform the best ISD's. This seems to point to the fact that it is precisely for this regime of parameters that we should aim for improving them. Interestingly enough, even if there is some room of improvement for RLPN decoding by using better strategies for producing the needed low weight parity-checks, there is a ceiling that this technique can not break (at least if we just split the parity-checks in two parts) and which is 
extremely close at rate $R=0.45$ to the best ISD algorithm \cite{BM18}. The RLPN decoding algorithm presented here has not succeeded in changing the landscape 
for very tiny code rates ($R$ going to $0$), since the complexity exponent of RLPN decoding approaches the one of exhaustive search on codewords, but the speed at which this complexity approaches exhaustive search is much smaller than for ISD's in the full decoding regime (\ie at the GV distance). The success of RLPN decoding for $R < 0.3$ could be traced back precisely to this behaviour close to $0$. An interesting venue for research could be to try to explore if there are other decoding strategies that would be candidate for beating exhaustive search in the tiny code rate regime. 

{Note however that like dual attacks in lattice based cryptography, the success of this algorithm relies on assumptions of the noise model we get from the low weight parity-check equations we produce (which is similar to the vectors in the dual lattice of small norm we use for dual attacks). The strict LPN model for this noise (Assumption \ref{ass:LPN}) has been found out not to be completely accurate for the large Fourier coefficients obtained during decoding the $[N,s]$-code with Fourier techniques (see Subsection \ref{ss:assumption}). However, 
a weaker conjecture, namely Conjecture \ref{con:main}, is enough for guaranteeing the success of this decoding method and is compatible with the experiments we have made. There is a rather clear path for verifying at least semi-heuristically this conjecture and this will be the object of further studies about this algorithm.}

\section*{Acknowledgement}We would like to express our warm gratitude to Elena Kirshanova and the Asiacrypt $22$' reviewers for their precious comments and remarks.
We wish also to thank Ilya Dumer for his very insightful thoughts about decoding linear codes in the low rate regime.

The work of TDA was funded by the French Agence Nationale de la Recherche through ANR JCJC COLA (ANR-21-CE39-0011). The work of Charles Meyer-Hilfiger was funded by the French Agence de l'innovation de d\'efense and by Inria.

\bibliographystyle{alpha}

\newpage
\newpage 
\appendix
\thispagestyle{empty}
\begin{center}
\vspace*{\stretch{1}}
{\Huge Appendices}
\vspace*{\stretch{1}}
\end{center}

\newpage 

\section{Asymptotic complexity of the RLPN decoder used in conjunction with the \cite{MMT11} and \cite{BJMM12} technique}
\label{sec:RLPN-MMT-BJMM}

\begin{proof}[Proof of Proposition \ref{pro:RLPNMMT}]

In the case of the MMT technique we have the following set of equalities 
\begin{equation}
2^{\ell_1}  =  \binom{w_2}{w_2/2}\binom{n-k+\ell-w_2}{p_1-w_2/2}
  \label{eq:representationMMT}
 \end{equation}
 \begin{equation}
S_0 =  \binom{\frac{n-k+\ell}{2}}{\frac{p}{2}}, \;\;\;
S_1 =  \frac{\binom{n-k+\ell}{p}}{2^{\ell_1}}, \;\;\;
S_2 =  \frac{\binom{n-k+\ell}{w_2}}{2^{\ell}}
\end{equation}
\begin{equation}
T_0 = S_0, \quad
T_1  =  S_0 + \frac{S_0^2}{2^{\ell_1}}, \quad 
T_2  =  S_1 + \frac{S_1^2}{2^{\ell-\ell_1}}, \quad 
\Teq  =  T_0+T_1+T_2
\end{equation}
We let 
$$
\sigma \eqdef \frac{s}{n}, \quad\lambda \eqdef \frac{\ell}{n}, \quad \lambda_1 \eqdef \frac{\ell_1}{n}, \quad \pi \eqdef \frac{p}{n}, \quad \omega_2 \eqdef \frac{w_2}{n}, \quad \nu_1 \eqdef \frac{v_1}{n}, \quad \nu_2 \eqdef \frac{v_2}{n}.
$$
The equality 
$$2^{\ell_1}  =  \binom{w_2}{w_2/2}\binom{n-k+\ell-w_2}{p_1-w_2/2}$$ translates into 
$$
\lambda_1  =   \omega_2 + (1-R+\lambda - \omega_2)h\left( \frac{\pi-\omega_2/2}{1-R+\lambda-\omega_2}\right).
$$
In our case, the number $N$ of parity-check which are produced is equal to $S_2$ and therefore the condition
$$
N = \Omt{\left(\frac{\binom{k-\ell-s}{w_1}\binom{n-k+\ell}{w_2}}{K_{w_1}^{k-\ell-s}(u_1)}\right)^2}.
$$
of Proposition \ref{prop:complexity3} becomes
\begin{multline*} 
(1-R+\lambda)h\!\left(\!\!\frac{\omega_2}{1-R+\lambda}\!\!\right)-\lambda = -\!(R-\lambda-\sigma)\tilde{\varepsilon}(\frac{\nu_1}{R-\lambda-\sigma},\frac{\omega_1}{R-\lambda-\sigma}) \\ 
-
(1-R+\lambda)\widetilde{\delta}(\frac{\nu_2}{1-R+\lambda},\frac{\omega_2}{1-R+\lambda}).
\end{multline*} 

\end{proof}
\begin{proposition}\label{pro:RLPN-BJMM}
The asymptotic complexity exponent of the RLPN decoder based on \cite{BJMM12} is given by
\begin{equation}
\alpha_{\text{BJMM}}(R) = \min_{(\sigma,\nu_1,\nu_2,\lambda,\lambda_1,\lambda_2,\pi_1,\pi_2,\omega_1,\omega_2) \in \sR} \beta(\sigma,\nu_1,\nu_2,\lambda,\lambda_1,\lambda_2,\pi_1,\pi_2,\omega_1,\omega_2)
\end{equation}
where
\begin{eqnarray*}
\beta & \eqdef & \min ( \sigma,\nu'+ \rho) \;\;\text{with} \\
\nu' &= & \max(\gamma_1,\gamma_2,\gamma_3), \quad \nu =(1-R+\lambda)h\left(\frac{\omega_2}{1-R+\lambda}\right)-\lambda, \; \;\\
\gamma_1 & = & \max\left(\frac{1-R+\lambda}{2}h\left(\frac{\pi_1}{1-R+\lambda} \right),(1-R+\lambda)h\left(\frac{\pi_1}{1-R+\lambda} \right)-\lambda_1\right)\\
\gamma_2 & = & 2(1-R+\lambda)h\left(\frac{\pi_1}{1-R+\lambda} \right)-\lambda_1-\lambda_2\\
\gamma_3 & = & 2(1-R+\lambda)h\left(\frac{\pi_2}{1-R+\lambda} \right)-\lambda - \lambda_2\\
\rho &= & 1-R -\sigma h\left( \frac{\tau-\nu_1-\nu_2}{\sigma}\right) 
-(R-\lambda-\sigma) h\left( \frac{\nu_1}{R-\lambda-\sigma} \right) -(1-R+\lambda) h \left( \frac{\nu_2}{1-R+\lambda}\right)\\
\tau & \eqdef & \deltaGV(R) = h^{-1}(1-R)
\end{eqnarray*}
and the constraint region $\sR$ is defined by the sub-region of nonnegative tuples $(\sigma,\nu_1,\nu_2,\lambda,\lambda_1,\lambda_2,\pi_1,\pi_2,\omega_1,\omega_2)$ such that
$$
	\sigma \leq R-\lambda, \quad \lambda_1 \leq \lambda_2 \leq \lambda, \pi_1 \leq \pi_2 \leq \omega_2, \quad  \nu_1 \leq R- \lambda - \sigma, \quad   \nu_2 \leq 1- R + \lambda, \quad 
	\tau- \sigma \leq \nu_1+\nu_2 \leq \tau, \quad \nu \leq \sigma,
$$
$$
\pi_1 \geq \frac{\pi_2}{2}, \quad \pi_2 \geq \frac{\omega_2}{2}, \quad 
\pi_2 < \lambda_1, \quad \omega_2 < \lambda_2, \quad
\pi_1 < 1-R + \lambda, \quad \omega_1 = \frac{R - \lambda - \sigma}{2} ,
$$
\begin{eqnarray*}
	\lambda_1 & = &  \pi_2 + (1-R+\lambda - \pi_2)h\left( \frac{\pi_1-\pi_2/2}{1- R+\lambda-\pi_2}\right),\\
	\lambda_2 & = &  \omega_2 + (1-R+\lambda - \omega_2)h\left( \frac{\pi_2-\omega_2/2}{1- R+\lambda-\omega_2}\right)
\end{eqnarray*}
and
$$
	\nu = -\!(R-\lambda-\sigma)\widetilde{\delta}(\frac{\nu_1}{R-\lambda-\sigma},\frac{\omega_1}{R-\lambda-\sigma})  
	-
	(1-R+\lambda)\widetilde{\delta}(\frac{\nu_2}{1-R+\lambda},\frac{\omega_2}{1-R+\lambda})
$$
where $\widetilde{\delta}$ is the function defined in Proposition \ref{prop:bias0}.

\end{proposition} \newpage
\section{About \cite{BM18}}
\label{sec:BM18}

We have compared all along the paper our results to the state-of-the-art for solving the decoding problem, that is \cite{BM18}. Actually, this paper claims better results than those presented in this paper (in particular in figures \ref{fig:compISDRLPN} and \ref{fig:lowerBound}). Indeed, in \cite{BM18}, the authors pretend that we can get a gain of $\approx$ 7\% for full decoding and at the worst code rate in comparison with \cite{BM17} (which is the state-of-the-art before \cite{BM18}).
However, this result is flawed: there is indeed an error  in the analysis of this new decoding algorithm which leads to this result. The decoding algorithm in \cite{BM18} essentially consists in producing lists recursively by using nearest-neighbor searches at each stage; the solution of the decoding problem is then contained in the last list. Thus, an accurate analysis of this algorithm lies in the good estimation of the size of the various lists. In \cite[Section 4, p. 16]{BM18}, an upper bound of the size of these lists is given. Unfortunately this upper bound does not hold.
Let us recall it. It is a bound of the list-size  $S_i$ at stage $i$:
\begin{equation}
	\label{eq:BM18}
		S_i \leq   \left|\{\  \mathbf{x} \in \mathbb{F}^k_2, |\mathbf{x}| = p_i\}\right|  \; \underset{\mathbf{x} \in \mathbb{F}^{\ell_i}_2}{\mathbb{P}} \left(|\mathbf{x}| = \omega_i^{(i)}\right) \displaystyle{\prod_{h=1}^{i-1}} \underset{\mathbf{x},\mathbf{y} \in \mathbb{F}^{\ell_h}_2}{\mathbb{P}} \left(|\mathbf{x}+\mathbf{y}| = \omega_h^{(i)} \ : \  |\mathbf{x}| = |\mathbf{y}| = \omega_h^{(i-1)}\right)
\end{equation}

We use here the notation from \cite{BM18}: at each stage $i \in \IInt{1}{m}$, the algorithm is producing lists of $S_i$ vectors with weight $p_i$ on an information set and weight $\left(\omega_h^{(i)} \right)_{h \in \IInt{1}{i}}$ on redundancy subsets of respective length $\left(\ell_i \right)_{h \in \IInt{1}{i}}$. Since the list at the final stage $m$ contains the solution, the decoding distance $\omega$ that we aim to achieve is $\omega = p_3 + \sum_{i = 1}^{m} \omega_i$.

We have noticed that there is a problem with \eqref{eq:BM18} by running simulations\footnote{The code in C for the aforementioned simulations can be found on \url{https://github.com/tillich/RLPNdecoding/tree/master/aboutBM18}.}. For instance, let look at the following parameters for depth $m = 3$:\\

\begin{center}
	\begin{tabular}{rcccccccc}
		Decoding problem parameters: & \ \ \ \ & $n = 120$ & \ \ \ & $k = 30$ & \ \ \ \  & $\omega = 42$ & \ \ \ \ & \ \ \ \  \\[0.2cm]
		Size of the redundancy subsets: & & $\ell_1 = 32$ & & $\ell_2 = 13$ &  & $\ell_3 = 36$ &  \\[0.2cm]
		Stage 1 parameters: &  & $p_1 =64$ & & $\omega_1^{(1)} = 12$ & & &  \\[0.3cm]
		Stage 2 parameters: & & $p_2 = 8$ & & $\omega_1^{(2)} = 24$ & & $\omega_2^{(2)} = 8$ & & \\[0.3cm]
		Stage 3 parameters: & & $p_3 = 10$ & & $\omega_1^{(3)} = 8$ & & $\omega_2^{(3)} = 10$ & & $\omega_3^{(3)} = 14$ 
	\end{tabular}
\end{center}

The Equation \eqref{eq:BM18} tells us that $S_2$ should not be bigger than $513$ whereas in practice it is around $2250$. In an other way, for those parameters, we have in addition that $S_3$ should be lower than $438$; which appears to be the case in practice but this time, the bound seems to not be tight at all since the lists at the last stage are actually of size around $8$. Finally, the upper bound \eqref{eq:BM18} seems to be at best not tight and at worst wrong.

The problem lies in the fact that in \eqref{eq:BM18}, the probability 
\begin{equation}
	P_i \eqdef \underset{\mathbf{x} \in \mathbb{F}^{\ell_i}_2}{\mathbb{P}} \left(|\mathbf{x}| = \omega_i^{(i)}\right)  \cdot \displaystyle{\prod_{h=1}^{i-1}} \underset{\mathbf{x},\mathbf{y} \in \mathbb{F}^{\ell_h}_2}{\mathbb{P}} \left(|\mathbf{x}+\mathbf{y}| = \omega_h^{(i)} \ : \  |\mathbf{x}| = |\mathbf{y}| = \omega_h^{(i-1)}\right)
\end{equation}
represents the probability that a pair from the lists at stage $i-1$ produces an element in the list at stage $i$, but it is not the probability that a vector of length $k$ and weight $p_i$ is in that list as it is suggested by  Equation \eqref{eq:BM18}. If we do not filter the duplicates, the actual expected size of the lists at stage $i > 1$ is $N_i P_i$ where
\begin{equation}
	N_i \eqdef S_{i-1}^2 \cdot \frac{\binom{k - p_{i-1}}{p_i/2}\binom{p_{i-1}}{p_i/2}}{\binom{k}{p_{i-1}}}
\end{equation}
is the number of pairs in a list at stage $i-1$ which are a representation of a vector of weight $p_i$. Note that we still have 
\begin{equation}
S_1 = \OO{\frac{\binom{k}{p_1} \binom{\ell_1}{\omega_1^{(1)}}}{2^{\ell_1}}}.
\end{equation}

After filtering the duplicates in the resulting lists, we finally obtain
\begin{equation}
	\label{eq:BM18goodformula}
	S_i = \OO{\min \left( N_i P_i\ , \ L_i \right)}
\end{equation}
where $L_i$ is the maximal size of the list, obtained when the whole set of vectors with the desired weight distribution is typically produced by the algorithm at the considered stage:
\begin{equation}
	\label{eq:BM18total}
	L_i  = \binom{k}{p_i} \prod_{h=1}^{i} \frac{\binom{\ell_h}{\omega_h^{(i)}} }{2^{\ell_h}}
\end{equation}

Note that if the assumptions in the correctness lemma \cite[Lemma 2, p.~17]{BM18} are met, then we actually have $S_{i} = L_i$ for all $i \in \IInt{1}{m}$.\\

In \cite[Theorem 3, Section 5, p.~19]{BM18} the following parameters are given for the full decoding at the hardest rate $\frac{k}{n} = 0.46$:
{\small
\begin{align*}
	& & \frac {\ell_1}{n} & =  0.0366 & \frac {\ell_2}{n} & = 0.0547  &  \frac {\ell_3}{n} &=  0.0911  &  \frac {\ell_4}{n} &=  0.3576  \\
	\frac{p_1}{n} &= 0.00559  & \frac{\omega_1^{(1)}}{n} &= 0.011515 & & & & & &  \\
	\frac{p_2}{n} &=  0.01073 & \frac{\omega_1^{(2)}}{n} &=0.023029 & \frac{\omega_2^{(2)}}{n}  &= 0.016676 & & & & \\
	\frac{p_3}{n} &= 0.02029 & \frac{\omega_1^{(3)}}{n} &=  0.0232 & \frac{\omega_2^{(3)}}{n}  &=0.033351 & \frac{\omega_3^{(3)}}{n}  &=  0.009993 & &   \\
	\frac{p_4}{n} &= 0.03460& \frac{\omega_1^{(4)}}{n} &=  0.0066 & \frac{\omega_2^{(4)}}{n}  &= 0.0099   & \frac{\omega_3^{(4)}}{n}  &= 0.0114 & \omega_4^{(4)}&=0.0612 
\end{align*}
}
With these parameters, if we believe the bound \eqref{eq:BM18}, we should have
\begin{align*}
	S_0 \leq 2^{0.02179 n} & & S_1 \leq 2^{0.03987 n} & & S_2 \leq 2^{0.05939 n}  & & S_3 \leq 2^{0.05975 n}
\end{align*}
But since the correctness lemma is verified we can use formula \eqref{eq:BM18total} to compute the expected size of the lists:
\begin{align*}
	S_0 = 2^{0.02179 n} & & S_1 = 2^{0.03987 n} & &S_2 = 2^{0.0655 n}  & & S_3= 2^{0.0705 n}
\end{align*}
We can see that $S_2$ and $S_3$ exceed their presumed upper bound, which increases the complexity of the algorithm to $\OO{2^{0.1083 n}}$ in comparison to the claimed one $\OO{2^{0.0885 n}}$.\\

As a result, we have re-optimized the Both-May algorithm by replacing the bound \eqref{eq:BM18} by the new formula \eqref{eq:BM18goodformula}. Actually, we have slightly modified the algorithm in \cite{BM18} by replacing the nearest-neighbor search routine, that stems from \cite{MO15}, by a more recent one that we can found in \cite{C20b} or \cite{EKZ21}:
\begin{theorem}[{\cite[Corollary 7.2.3, p.~183]{C20b}} or {\cite[Theorem 1, p.~8]{EKZ21}}]
For any constants $\lambda \in \left[0,1\right]$ and $\omega \in \left[0,\frac{1}{2}\right]$, when $n$ tends to infinity, the time complexity for finding all the pairs (except $\oo{1}$ of them) of binary vectors at distance $\lfloor \omega n \rfloor$ in a list of $\OO{2^{\lambda n}}$ vectors of length $n$ is
\begin{equation}
\mathrm{NNS}\left( n, \lfloor \omega n \rfloor, 2^{\lambda n} \right)  = \OOt{2^{\alpha n}}
\end{equation}
where
\begin{equation*}
\alpha = \left\lbrace 
\begin{array}{ll}
(1 - \omega)\left( 1 - h\left( \frac{h^{-1}(1 - \lambda) - \omega/2}{1 - \omega}\right)\right) \ \ \ \ \ & \text{if } \frac{1 - \sqrt{1 - 2\omega}}{2} < h^{-1}(1 - \lambda)\\
2\lambda - 1 + h(\omega) & otherwise
\end{array}
\right.
\end{equation*}
\end{theorem}

Armed with this tool and considering the new estimation of the list sizes, the corrected time complexity of the Both-May decoder is then
\begin{equation}
T_{\mathrm{BM18}} = \OOt{\dfrac{\max_{i \in \IInt{1}{m}}\left( T_i \right)}{P_{\mathrm{succ}}}}
\end{equation}
where $T_i \eqdef \mathrm{NNS}\left(\ell_i, \omega_i^{(i)}, S_{i-1}\right)$ is the cost for producing the lists at the stage $i$ and $P_{\mathrm{succ}}$ is the probability of success of an iteration of the Both-May algorithm. When the conditions of the correctness lemma \cite[Lemma 2, p.~17]{BM18} are met, we have
\begin{equation}
P_{\mathrm{succ}} = \dfrac{\binom{k}{p_m} \prod_{i = 1}^{m} \binom{\ell_i}{\omega_i^{(m)}}}{\binom{n}{t}}
\end{equation}

Finally, Figure \ref{fig:BothMay} illustrates the results we obtained with our new analysis of \cite{BM18}. In this figure, the corrected complexity of Both-May algorithm is indistinguishable from the complexity of the \cite{BM17} decoder. More precisely, the first one is slightly better than the second; in particular, the optimized complexity for the full decoding problem is $\OO{2^{0.0953n}}$ with the \cite{BM17} decoder at the hardest rate $R = 0.46$ and it is $\OO{2^{0.0950n}}$ with \cite{BM18} at the hardest rate $R = 0.424$.\footnote{The code in C++ for optimizing the corrected \cite{BM18} complexity and some tables containing the optimized parameters for full decoding at various rates are given on \url{https://github.com/tillich/RLPNdecoding/tree/master/aboutBM18}}

\begin{figure}[h!]
	\centering
	\includegraphics[height=6.2cm]{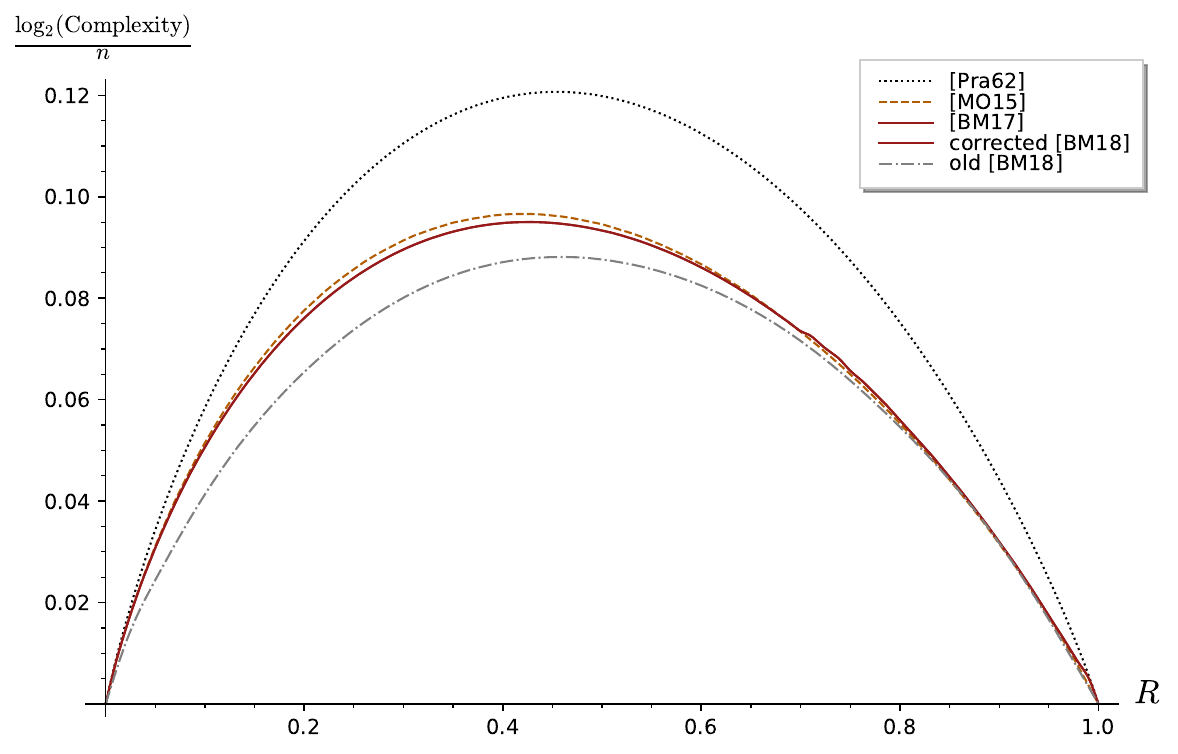}

	\caption{\label{fig:BothMay} Corrected complexity of the \cite{BM18} decoder.}
\end{figure} 
\end{document}